\documentclass[aps,pra,twocolumn,groupedaddress]{revtex4-2}
\usepackage[dvipsnames]{xcolor}
\usepackage{hyperref}
\hypersetup{colorlinks=true, linkcolor=myblue, citecolor=ForestGreen, urlcolor=DarkOrchid}
\usepackage[a4paper,left=2.5cm,right=2cm,top=2.5cm,bottom=3cm,headheight=2cm]{geometry} 
\usepackage{graphicx}
\usepackage{amsmath}
\usepackage{amsthm}
\usepackage{amssymb}
\usepackage{bbm}
\usepackage{enumitem}
\usepackage{tcolorbox}
\tcbuselibrary{listings,theorems}
\usepackage{algorithm}
\usepackage{varwidth}
\usepackage[noEnd=true,indLines=false]{algpseudocodex}
\makeatletter
\@addtoreset{ALG@line}{algorithm}
\makeatother
\usepackage[nolist]{acronym}
\usepackage{csquotes}
\usepackage[normalem]{ulem}
\usepackage{cancel}
\usepackage[]{todonotes}
\usepackage{booktabs} 
\usepackage{pgfplotstable} 
\usepackage{dsfont}
\allowdisplaybreaks[1]
\usepackage{mathtools}
\mathtoolsset{showonlyrefs=true} \usepackage{physics}
\usepackage{pgfplots}
\pgfplotsset{compat=1.3}
\usepgfplotslibrary{statistics}
\usepgfplotslibrary{fillbetween}
\usepgfplotslibrary{groupplots}
\usepackage{ragged2e}
\usepackage{float}
\usepackage{tikz}
\usetikzlibrary{
    calc, 
    decorations,
    positioning,
    arrows.meta,
    intersections, 
    backgrounds,
    shadows.blur,
    fit,
    decorations.pathreplacing,
    shapes,
    matrix,
    math,
    quantikz2,
    graphs,
    external
}

 \newcommand{\ii}{\mathrm{i}}

\newcommand{\ee}{\mathrm{e}}

\makeatletter
\DeclareRobustCommand{\rvdots}{\vbox{
    \baselineskip4\p@\lineskiplimit\z@
    \kern-\p@
    \hbox{.}\hbox{.}\hbox{.}
  }}
\makeatother

\makeatletter
    \tikzset{external/shell escape={-shell-escape\space-output-directory=build}
    }
\makeatother

\tikzexternaldisable
\newcommand{\externalize}[2]{
  \tikzsetnextfilename{#1}{#2}\tikzexternaldisable
}

\newtheorem{theorem}{Theorem}
\newtheorem{lemma}[theorem]{Lemma}
\newtheorem{prop}[theorem]{Proposition}
\newtheorem{cor}[theorem]{Corollary}
\newtheorem{rem}{Remark}

 \definecolor{sbblue}{rgb}{0.2823529411764706, 0.47058823529411764, 0.8156862745098039}
\definecolor{sborange}{rgb}{0.9333333333333333, 0.5215686274509804, 0.2901960784313726}
\definecolor{sbgreen}{rgb}{0.41568627450980394, 0.8, 0.39215686274509803}
\definecolor{sbred}{rgb}{0.8392156862745098, 0.37254901960784315, 0.37254901960784315}
\definecolor{sbpurple}{rgb}{0.5843137254901961, 0.4235294117647059, 0.7058823529411765}
\definecolor{sbbrown}{rgb}{0.5490196078431373, 0.3803921568627451, 0.23529411764705882}
\definecolor{sbmagenta}{rgb}{0.8627450980392157, 0.49411764705882355, 0.7529411764705882}
\definecolor{sbgray}{rgb}{0.4745098039215686, 0.4745098039215686, 0.4745098039215686}
\definecolor{sbocca}{rgb}{0.8352941176470589, 0.7333333333333333, 0.403921568627451}
\definecolor{sblightblue}{rgb}{0.5098039215686274, 0.7764705882352941, 0.8862745098039215}
\definecolor{sbbackground}{RGB}{234,234,242}

\definecolor{lightgray}{rgb}{0.7,0.7,0.7}
\definecolor{darkgreen}{rgb}{0,0.39,0}
\definecolor{darkblue}{rgb}{0,0,.6}
\definecolor{darkred}{rgb}{.6,0,0}

\definecolor{myblue}{RGB}{2,61,107}
 
\begin{document}
\title{Measurement-driven Quantum Approximate Optimization
}
\author{Tobias Stollenwerk}
\affiliation{Institute for Quantum Computing Analytics (PGI-12), J\"ulich Research Centre, Wilhelm-Johnen-Stra\ss e, 52428 J\"ulich, Germany} 
\email{to.stollenwerk@fz-juelich.de}
\author{Stuart Hadfield}
\affiliation{Quantum Artificial Intelligence Lab (QuAIL), NASA Ames Research Center, Moffett Field, CA 94035, USA}
\affiliation{USRA Research Institute for Advanced Computer Science (RIACS), Mountain View, CA 94043, USA}
\email{stuart.hadfield@nasa.gov}
 \date{\today}
\begin{abstract}
    Algorithms based on non-unitary evolution have attracted much interest for ground state preparation on quantum computers.  
One recently proposed method makes use of ancilla qubits and controlled unitary operators to implement weak measurements related to imaginary-time evolution.
In this work we specialize and extend this approach to the setting of combinatorial optimization.
We first generalize the algorithm from exact to approximate optimization, taking advantage of several properties unique to 
classical problems. In particular we show how to select parameters such that the success probability of each measurement step is 
bounded away from $1/2$.
We then show how to 
adapt our paradigm to the setting of constrained optimization for a number of important classes of hard problem constraints. 
For this we compare and contrast both penalty-based and feasibility-preserving approaches, elucidating the significant advantages of the latter approach.
Our approach is general and may be applied to easy-to-prepare initial states as a standalone algorithm, or deployed as a quantum postprocessing stage to improve performance of a given parameterized quantum circuit.
We then propose a more sophisticated variant of our algorithm that adaptively applies a mixing operator or not, based on the measurement outcomes seen so far, as to speeds up the algorithm and helps the system evolution avoid slowing down or getting stuck suboptimally. 
In particular, we show that mixing operators from the quantum alternating operator ansatz 
can
be imported directly, both for the necessary eigenstate scrambling operator and for initial state preparation, and discuss quantum resource tradeoffs.  
 \end{abstract}
\maketitle

\begin{acronym}
\acro{algosimple}[MDQO]{Measurement-driven Quantum Optimization}
\acro{algo}[FC-MDQO]{Feedback-controlled Measurement-driven Quantum Optimization}
\end{acronym}
 
\section{Introduction}
Combinatorial optimization remains a broad and computationally challenging application area that has attracted much interest as an attractive target for potential quantum advantages in the near term and beyond~\cite{abbas2024challenges,kadowaki1998quantum,farhi2000quantum,albash2018adiabatic,farhi2014quantum,hadfield2019quantum,
sanders2020compilation}. 
While a 
variety of quantum algorithms have been proposed in this domain based on coherent unitary evolution, 
these algorithms come with a number of practical and theoretical shortcomings or difficulties, at least in terms of achieving advantages with real-world quantum devices of the foreseeable future~\cite{stilck2021limitations}. 
For instance, variational quantum algorithms require a parameter optimization stage, which can be as computational difficult as the target problem itself~\cite{bittel2021training}, even at short depths.
Other approaches can avoid this problem in some regimes but typically come with performance tradeoffs~\cite{abbas2024challenges}.   

In light of these and other difficulties, recent work had begun to explore more sophisticated approaches that use the quantum device in novel ways.
In particular, 
approaches based on non-unitary quantum evolution have attracted much interest for ground state preparation on quantum computers~\cite{motta2020determining,cubitt2023dissipative,mao2023measurement,ding2024single,bauer2024combinatorial,jordan2024optimization,choi2021rodeo,gustafson2020projective,hubisz2020quantum}. 
While these methods provide no free lunch, and naturally come with there own difficulties and tradeoffs, they present a new avenue towards practically useful 
approaches to 
quantum optimization.

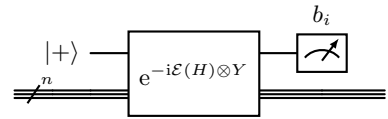
\begin{figure}[htb]
    \begin{center}
        \externalize{circuit_teaser_mdqo}{\begin{quantikz}[wire types={q,b}, classical gap=0.05cm, row sep={4ex, between origins}]
    \setwiretype{n}
    &
    & \lstick{$\ket{+}$}
    & \gate[2]{\ee^{-\ii \mathcal{E}(H) \otimes Y}} \setwiretype{q}
    & \meter{b_i} 
    & \setwiretype{n}
    &
    \\
    & \qwbundle{n}
    &
    &
    &
    &
\end{quantikz}

         }
    \end{center}
    \caption{Basic weak-measurement step for suitably transformed cost Hamiltonian $\mathcal{E}(H)$. The bottom system register encodes a distribution over problem solutions.}\label{fig:teaser}
\end{figure}

At the same time, quantum hardware continues to develop at steady pace, across of variety of platforms and architectures.
An exciting avenue is the advent of new hardware features such as mid-circuit measurement and adaptive control~\cite{chiaverini2004realization,botelho2022error,singh2023mid,koh2024readout,bar2025layered,govia2023randomized,hothem2025measuring}, inspired by quantum error correction applications as well as algorithmic ones. 
One recently proposed method~\cite{mao2023measurement} for ground state preparation 
makes use of mid-circuit 
weak measurements, with outcomes controlling subsequent circuit elements,  
to affect dynamics related to imaginary-time evolution. 
Each weak measurement is implemented through introducing an ancilla qubit, applying an entangling gate with the system register, and then measuring the ancilla as is depicted in Figure~\ref{fig:teaser}. 
Given a cost Hamiltonian~$H$ to optimize, the protocol of~\cite{mao2023measurement} utilizes an entangling gate derived using 
$\mathcal{E}(H) = \varepsilon H$ for a suitable 
constant $\varepsilon$;  
we will advantageously consider more general Hamiltonian transformations, as well as careful derivation of the transformation parameters to yield improved performance. 
Each weak measurement transforms the quantum state in a way that depends on 
$H$, and the set of outcomes obtained from repeated weak measurements can then be used to infer properties of the 
resulting quantum state, specifically regarding the induced distribution over candidate problem solutions.
Based on these properties one can adaptively decide to perform further weak measurement steps, or to stop, which for the case of combinatorial optimization amounts to measurement of the system register which returns a candidate solution sample.
When fast feedforward control~\cite{bluvstein2024logical,acharya2025quantum} is available, classically controlled quantum circuit elements may be adaptively applied in between weak measurements.
In each setting the goal remains the same, to drive the system register to an improved distribution over problem solutions relative to the initial state.

\subsection{Our approach and contributions} 
We adapt and 
extend the basic weak-measurement scheme of \cite{mao2023measurement} to the setting of combinatorial optimization, with a number of adaptions specific to this setting.
As optimization problems of interest are typically NP-hard, we do not expect to efficiently solve them  on quantum devices with any algorithm, so we first specialize to the setting of approximate optimization.
Here we first propose adaptive low-depth implementations that are designed to be repeated many times, with each run yielding a sample candidate solution. 
We explain how to transform the cost Hamiltonian such that desirable sequences of weak measurement outcomes are likely to occur, though as expected we show tradeoffs exist 
between the improvement per step and the 
corresponding \emph{success probability} of obtaining the desired weak measurement outcome.
In particular, 
by carefully setting algorithm parameters we can guarantee that the success probability of each weak measurement step 
is lower bounded away from $1/2$.
We later give a more general algorithm suitable for approximate or exact optimization in the context of (increasingly) fault-tolerant quantum devices. 

Our algorithms  
trade additional resources to 
probabilistically output samples from 
improved 
distributions over problem solution relative to the input state.
They may be applied to an easy-to-prepare initial state, such as a product state, or to an improved initial state at the expense of more resources, such as a pre-optimized parameterized quantum circuit. 
The latter case allows for including wide classes of quantum algorithms, variational or otherwise, such as QAOA~\cite{farhi2014quantum,hadfield2019quantum} and its variants, among others~\cite{abbas2024challenges}. As the quality (i.e., cost expectation) of the input state improves, so does the success probability of each step of our algorithm. 

Given a cost Hamiltonian $H$ to optimize, the core of our approach is to implement a sequence of weak measurements 
(cf.~Figure~\ref{fig:teaser}) 
that transform an initial state $\ket{\psi_0}$ to one proportional to 
\begin{equation}
    \sin^{k_1} (\tfrac{\pi}4+\mathcal{E}(H)) \cos^{k_0} (\tfrac{\pi}4+\mathcal{E}(H)) \; \ket{\psi_0}  
\end{equation}
where $k_0,k_1$ are the numbers of $0,1$ outcomes of the $k=k_0+k_1$ measurements, respectively. 
As we will show, the $1$ outcome improves the cost expectation of the previous state and so is considered a \enquote{success}. 
An important innovation we introduce is shifting and rescaling of the cost Hamiltonian,
$\mathcal{E}(H)$, a suitably chosen transformation of $H$, for example a polynomial.  
For simplicity we will primarily focus on linear transformations $\mathcal{E}(H)=\epsilon (\alpha I+ H)$.
We show that the constants $\alpha,\epsilon$ can be selected as to improve the success probabilities of the weak measurements (with respect to a naive application of the approach of \cite{mao2023measurement}). 
Indeed at any given step with current state~$\ket{\psi}$, the weak measurement success probability~$p_1$ is
shown to be \begin{equation} 
    p_1 
= \frac{1}{2} + \frac{1}{2} \bra{\psi}  \sin (2 \mathcal{E}(H)) \ket{\psi} \; > \frac12 + \text{const}
\end{equation}
when $\alpha,\epsilon$ derived appropriately from the particular problem instance; 
in particular we show it suffices to select parameters such that $0\leq \mathcal{E}(H)\leq \pi/4$. 
Hence we then have the important result that 
the success probability 
can be guaranteed to be strictly greater than $1/2$ for all possible states $\ket{\psi}$, 
which implies this property holds from all possible sequence of states resulting from repeated weak measurements applied to a given input state.
While this result holds generically, we explain how further improvements are possible in the case when one has additional knowledge regarding the support of the initial or current state. 
In practice one does not know the spectrum of $H$ by definition and so bounds can be used to select suitable parameters.  
We further discuss and analyze how the choice of these constants affects performance and show tradeoffs between success probability and relative improvement to the cost expectation, 
with the counterintuitive result that using (apparently suboptimal) bounds can sometimes improve performance.

Our algorithms are meant to be run many times, with the number of weak measurement steps in each run tailorable to the capabilties of current and future quantum hardware. 
It is hence advantageous to consider different notions of 
what defines a successful run
(not to be confused with \enquote{success} of each weak measurement step),
as well as different algorithm variants, particularly different return and overall termination criteria.  
When a target surplus or fraction of successes is achieved, corresponding to a desirable output distribution, the procedure terminates in that the system register is measured and a solution candidate returned.
Then the process is restarted.
For "bad" sequences of runs with a significant number of unsuccessful outcomes, the current run can be terminated (i.e., system register measured) early and the process restarted.
To this end the algorithm is embedded in an iterative outer control loop that decides when to proceed or restart based on the measurement data obtained so far, as indicated schematically in Figure~\ref{fig:flow-general}. 
Furthermore the parameters can be adaptively updated based on the overall set of outcomes obtained. For example, once a solution of given quality obtained, algorithm meta-parameters can be adjusted as to target doing even better. 
\begin{figure}[H]
    \begin{center}
        \externalize{flow_chart_general}{\tikzset{base/.style={draw=none,
        align=center,
        text centered,
    },
    box/.style={base,
        rectangle,  
        rounded corners, 
        fill=sbblue!50!white,
    },
    condition/.style={base,
        draw=none,
        diamond,
        fill=sbgreen!50!white,
        rounded corners,
    },
    measurement/.style={base,
        trapezium, 
        trapezium left angle = 65,
        trapezium right angle = 115,
        trapezium stretches,
        fill=sbred!50!white,
        rounded corners,
    },
    arr/.style={->,
        line width=1pt,
        color=black!90!white,
        -latex,
    },
    line/.style={line width=1pt,
        color=lightgray,
    },
    mylabel/.style={pos=0.5,
        fill=white,
        inner sep=0.5ex,
    },
}

\begin{tikzpicture}[font=\small,thick]

\node (qopt) [box, 
    minimum width=15em,
    minimum height=29ex,
    inner sep=1.0ex, 
    align=center,
    text depth=24ex,
    ]{Algorithm~1 or Algorithm~2
    };
\node (init) [box, 
    above=4ex of qopt,
    align=center,
    text width=7em,
    ]{Prepare initial state
    };
\node (weakstep) [box,
    inner sep=0.2ex,
    fill=white!80!sbblue,
    minimum height=22ex,
    text width=8em,
    xshift=1ex,
    yshift=-2ex,
    anchor=west,
    ] at (qopt.west) {weak\\measurement\\step\\(repeated)
        \\
        \includegraphics[width=1.0\textwidth]{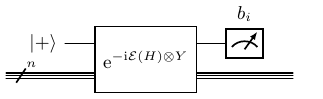}
    };
\draw[arr, rounded corners] ([yshift=3ex]weakstep.south east)
    -- ++(3.2em, 0ex)
    node[mylabel, 
        fill=white!50!sbblue,
        pos=0.4]{$b$}
    |- 
    node (midlabel) [box, 
        inner xsep=0.0pt,
        text width=6em,
        pos=0.25]{measurement outcome}
    ([yshift=-3ex]weakstep.north east)
    ;

\node (result) [box,
    fill=white,
    text width=6em,
    minimum width=5em,
    below=5ex of qopt,
    ] {measured bitstring $\mathbf{x}$
    };
\node (outerwhile) [condition,
    yshift=-17ex,
    inner sep=0.3ex,
    right=1em of qopt,
    ] {Termination
        \\
        criteria $\mathcal{T}(\mathbf{x})$
    };
\node (endresult) [box,
    fill=white,
    text width=6em,
    minimum width=5em,
    below=8ex of outerwhile,
    ] {resulting bitstring $\mathbf{y}$
    };

\draw[arr] (qopt) -- (result);
\draw[arr, rounded corners] (result.south)
    -- ++(0, -2ex)
    -- ++(7em, 0)
    |- (outerwhile.west)
    ;
\draw[arr, rounded corners] (outerwhile.north)
    -- ++(0, 35ex)
    node[mylabel, pos=0.2]{No}
    node[box, 
        pos=0.50]{(Adapt\\parameters?)}
    -| ([xshift=1em]init.north)
    ;
\draw[arr]  ([yshift=3ex,xshift=-2em]init.north) -- node[above, text width=5em, align=center, pos=0.0]{initial parameters}([xshift=-2em]init.north) ;
\draw[arr] (outerwhile) -- (endresult) node[mylabel, pos=0.4]{Yes};
\draw[arr] (init) -- (qopt);

\end{tikzpicture}

         }
    \end{center}
    \caption{Schematic depiction of the overall procedure. The blue box depicts either of Algorithms~\ref{alg:mdqo} and~\ref{alg:fcmdqo} given below, embedded in an outer control loop that reruns the algorithms, possibly adjusting their parameters based on the information obtained so far. When a solution meeting a target value is obtained (or other criteria, for example time limit reached), the overall best solution found~$\mathbf{y}$ is returned. 
}\label{fig:flow-general}
\end{figure}
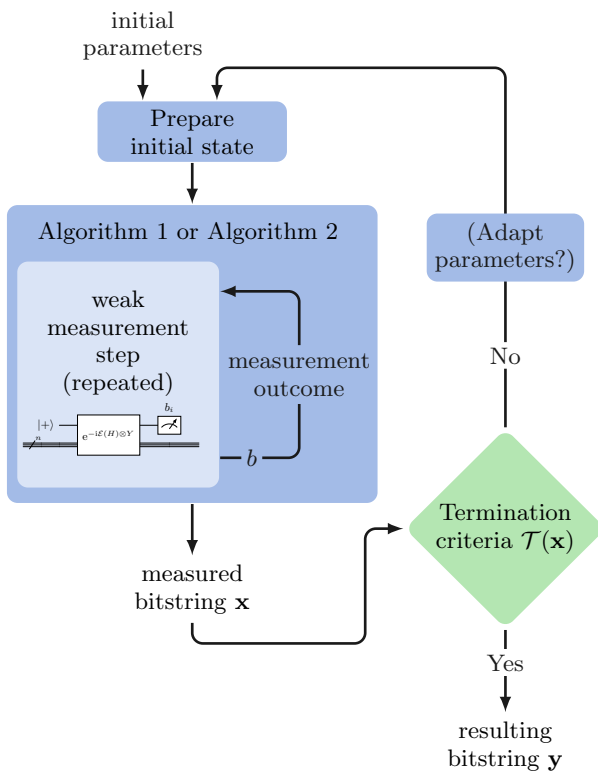

We first propose and analyze Algorithm~\ref{alg:mdqo}, which consists of repeated weak measurement steps, with a control loop that determines when to halt to the process based on the outcomes obtained so far. We analyze Algorithm~\ref{alg:mdqo} to show that while successful steps usually improve the output distribution, the procedure can get stuck for example if the state becomes too close to a cost eigenstate, or if given such a state as input. To overcome such challenges  
we generalize to Algorithm~\ref{alg:fcmdqo}, which uses the weak measurement outcomes in a more sophisticated way to control aperiodic applications of a "scrambling" operator which by design induces amplitude exchange between cost eigenstates.
For both Algorithms~\ref{alg:mdqo} and~\ref{alg:fcmdqo} we explain how each may be directly extended to the important and more general setting of constrained optimization problems. 
In particular
we show that algorithmic primitives from quantum alternating operator ansatze (QAOA)~\cite{farhi2014quantum,hadfield2019quantum} can be directly imported 
to yield suitable scrambling operators for wide classes of problems, and that comparable advantages are often obtained in our setting over other approaches such as the use of penalty terms. 
Finally, we provide numerical demonstrations of our method for both the MaxCut and Max Independent Set problems, with the same graph instance used throughout the paper for convenience. 

\paragraph*{Quantum post-processing:}
Importantly we propose two distinct settings of application for quantum optimization algorithms based on weak measurements, i.e. for Algorithms~\ref{alg:mdqo} and~\ref{alg:fcmdqo}. 
The first considers a fiducial initial state such as the equal superposition state $\ket{\psi_0}=\ket{+}^{\otimes n}=\sum_x\tfrac1{\sqrt{2^n}} \ket{x}$, or another that is relatively easy to prepare, and as such yields a self-contained algorithm.
The second setting envisons Algorithm~\ref{alg:mdqo} instead as a \textit{quantum post-processing} procedure.
Here $\ket{\psi_0}$ is some input state which may depend on the problem $H$, created by a nontrivial quantum algorithm, such as a quantum circuit ansatz. 
For example, $\ket{\psi_0}$ could be a QAOA$_p$ state for $H$, or possibly another ansatz.
We show that even applying a single weak measurement step as post-processing produces samples from an improved distribution as compared to the ansatz alone, when post-selecting on the successful outcome.

\paragraph*{Zero-cost constraint preservation:} 
Broadly speaking, quantum algorithms for constrained optimization problems deal with constraints by biasing dynamics towards feasible states through penalty terms, or directly implementing operators that are constraint invariant on feasible initial states. Both approaches typically come with significant overhead, which scales with the circuit depth. We show that for Algorithm~\ref{alg:mdqo}, the only requirement is that the initial state is feasible. As the action of the weak measurement operator is diagonal with respect to the computational basis, this guarantees that subsequent system states remain feasible independent of the specific measurement outcomes. Hence for broad classes of constrained optimization problems the overhead from the constraints is independent of the number of weak measurement steps. 

The remainder of the paper is structured as follows.
The following subsection briefly overviews connections between  weak-measurement based algorithms and a number of related protocols in the literature. 
We detail and analyze Algorithm~\ref{alg:mdqo} in Section~\ref{sec:mdqo}, and explicitly generalize it to constrained optimization in Section~\ref{sec:constrainedOpt}. 
In Section~\ref{sec:fcmdqo} 
we give Algorithm~\ref{alg:fcmdqo} which generalizes  Algorithm~\ref{alg:mdqo} for the case that fast mid-circuit feedback is available. 
Finally, we discuss several extensions of our algorithms and provide a future outlook in Section~\ref{sec:conclusion}.

\subsection{Related approaches} \label{sec:relatedApproaches}
Algorithms based on iterated matrix or operator multiplication are a ubiquitous paradigm across computational sciences, especially for eigenvalue and eigenvector estimation problems. 
Classically, the most fundamental example is the power method.
Given a diagonalizable matrix $A$, a suitable initial state $\vec{\psi}$, and an integer $k$, one computes $A^k\vec{\psi}$, which converges with $k$ to an eigenvector of $A$ corresponding to the eigenvalue of largest absolute value.
A number of more sophisticated methods extend the same basic idea with transformed matrices, such as inverse iteration, or classical methods related to imaginary time evolution.

In the quantum case our approach aims to \enquote{multiply} (up to, of course, state normalization) the initial state by a trigonometric polynomial of the cost Hamiltonian $H$ as in Equation~\eqref{eq:ent-meas-operation}. 
A sequence of recent papers~\cite{rodriguez2017heat,motta2020determining,mao2023measurement,alam2023solving,bauer2024combinatorial,morris2024performant,gluza2024double,zhong2024classical,wang2025imaginary} have proposed quantum algorithms directly based on or related to imaginary time evolution, which correspond to creating a quantum state proportional to $\exp(\tau H)\ket{\psi_0}$ with $\tau \in \mathds{R}$.\footnote{We emphasize that if one could implement arbitrary imaginary time evolution
then one can solve hard computational problems such as those known to be NP-hard, a task we don't believe to be generally efficient on a quantum computer.
Hence \say{imaginary time} approaches typically trade this complexity into other aspects such requiring a large number of measurements or effectively state tomography, difficulty in parameter setting, or 
reduced  
success probability.}
At a high-level, these approaches like our own try to drive a system register toward a desired state or subspace.

More broadly, numerous authors have derived algorithms within this theme of dynamics derived from or inspired by ideas in physics.
These include engineered approaches related to exploiting effects such as dissipation, damping, and cooling, e.g \cite{verstraete2009quantum,polla2021quantum,ding2024single,basso2024optimizing,cubitt2023dissipative,zhang2023dissipative,zapusek2025scaling,zhan2025rapid,maciejewski2024improving}, among others. 
Our approach is also closely related to ideas from measurement induced quantum steering~\cite{roy2020measurement,volya2025state}.
In terms of utilizing measurements as algorithmic primitives, 
a long line of closely related work has considered 
algorithms driven by project measurements, as well as approaches in the measurement-based quantum computational model, to both optimization and decision problems~\cite{childs2002quantum,benjamin2017measurement,zhao2019measurement,ferguson2021measurement,mao2023measurement,stollenwerk2024measurement,zhang2024solving}. 
Our approach of repeated \lq\lq weak\rq\rq\ measurements amounts is also naturally related to approaches based on the adiabatic theorem as well as the quantum Zeno effect, themselves both successful algorithmic primitives~\cite{farhi2000quantum,albash2018adiabatic,papageorgiou2014estimating,kremenetski2023quantum,herman2023constrained,berwald2025zeno}.
In particular our choice of weak measurement is closely connected to 
though distinct from (repeated) single-qubit Quantum Phase Estimation~\cite{kitaev2002classical,
dutkiewicz2022heisenberg,ni2023low}.

Finally we remark that a recently proposed paradigm called Decoded Quantum Interferometry~\cite{jordan2024optimization,schmidhuber2025hamiltonian,khattar2025verifiable}, attempts to create states proportional to $p(H)\;\ket{+}^{\otimes n}$, for $p(\cdot)$ a suitably chosen polynomial.
The catch is that for this algorithm to succeed the quantum computer must also coherently solve a reduced problem related to decoding, which is nontrivial. 
Our approach instead seeks to create a state proportional to a trigonometric polynomial of a suitably rescaled and shifted cost Hamiltonian $C=\mathcal{E}(H)$, which avoids the need for complicated subroutines to be executed coherently as well as other resource overheads. 
 
\section{Measurement-driven Quantum Optimization}
\label{sec:mdqo}
In contrast to general settings where one may seek to prepare approximate \emph{quantum} eigenvectors (i.e., superpositions of basis states), 
in combinatorial optimization we seek \emph{classical} bit strings that give high-quality solutions, with each  
corresponding to individual computational basis states. 
Hence classical candidate solutions can always be efficiently checked and stored, which is not true for the quantum case generally. 
While the optimal solution is always desirable, due to computational complexity considerations a high-quality approximation is often all we can hope to efficiently achieve, and as such we focus on the approximate optimization setting. Here we follow the general recipe of repeated state preparation and measurement common to quantum optimization algorithms, which produces a large number of candidate solution samples, and the best solution found overall is returned; refer to the outer loop of Figure~\ref{fig:flow-general}. 
We leverage various properties particular to the combinatorial optimization setting in the design of our algorithms. 

Here we propose our first quantum algorithm for approximate optimization based on weak measurements: \ac{algosimple}, given in Algorithm~\ref{alg:mdqo} and depicted in Figure~\ref{fig:flow-mdqo}. 
It is inspired by the approach of~\cite{mao2023measurement} for quantum ground state estimation, but tailored to approximate optimization. 
The algorithm is potentially attractive for the near-term, especially 
when relatively few weak measurement steps are to be applied before termination. 
For simplicity of presentation we focus here on the case of unconstrained optimization problems; we extend our approach to the more general setting of constrained optimization in Section~\ref{sec:constrainedOpt}. 

As explained Algorithm \ref{alg:mdqo} is to be run repeatedly as indicated in Figure~\ref{fig:flow-general}, after which the best solution found overall is returned.

\begin{algorithm}[H]
    \textbf{Input:}
\begin{itemize}[leftmargin=2em]
\item 
        Problem Hamiltonian $H$ encoding a pseudo-boolean cost function $h(x)$ to maximize over $x\in\{0,1\}^n$
    \item
        Rescaled Hamiltonian $C:=\mathcal{E}(H)$ satisfying\\ $0\leq C \leq \frac{\pi}4$
    \item 
        Initial $n$-qubit state  $\ket{\psi_0}$ 
    \item 
        Return criteria 
        $\mathcal{R}(\mathbf{b},\dots)$
        depending on measurement outcomes $\mathbf{b}$ 
        and other parameters 
\end{itemize}
\textbf{Output:} 
\begin{itemize}[leftmargin=2em]
    \item 
        Classical bitstring giving approximate solution
\end{itemize}
\vskip 0.2pc
\begin{algorithmic}[1]
    \State $\ket{\psi} \gets \ket{\psi_0}$ 
    \State $\mathbf{b} \gets \emptyset$
    \While{$\neg \; \mathcal{R}(\mathbf{b})$} 
        \State Introduce fresh ancilla qubit in $\ket{+}$ state
            \label{alg:line:mdqo_fresh_ancilla}
        \State 
            \begin{varwidth}[t]{0.9\textwidth}
                Apply entangling gate across system and ancilla registers 
                \begin{equation}
                    \ket{\mathbf{\Psi}} \gets \ee^{-\ii C \otimes Y}\ket{\psi} \otimes \ket{+}
                \end{equation}
            \end{varwidth}
            \label{alg:line:mdqo_entangling}
        \State 
            \begin{varwidth}[t]{0.9\textwidth}
                Measure ancilla qubit in computational basis with outcome $b \in \{0, 1\}$ \label{alg::line:mdqo_measurement}
            \end{varwidth}
            \label{alg:line:mdqo_measurement}
        \State Append $b$ to $\mathbf{b}$
    \EndWhile 
    \State 
        \begin{varwidth}[t]{0.9\textwidth}
            Measure $\ket{\psi}$ in computational basis and return the resulting $n$-bit string
        \end{varwidth}
\end{algorithmic}

     \caption{\acf{algosimple}} \label{alg:mdqo}
\end{algorithm}
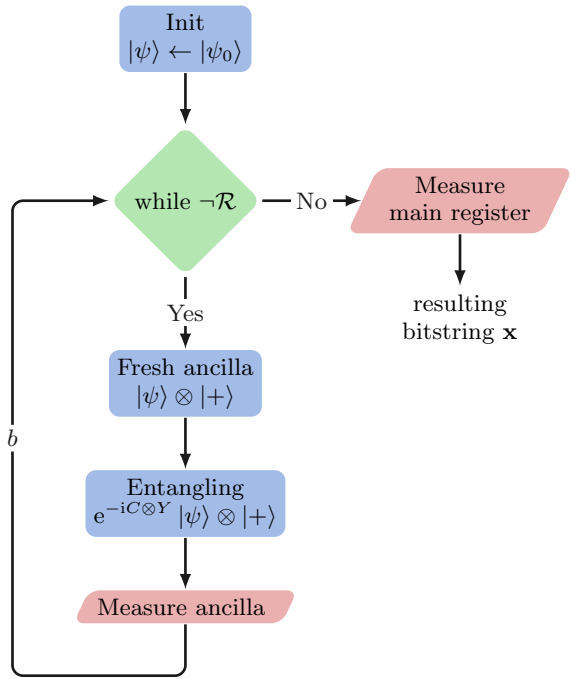
\begin{figure}
    \begin{center}
        \externalize{flow_chart_mdqo}{\tikzset{base/.style={draw=none,
        align=center,
        text centered,
    },
    box/.style={base,
        rectangle,  
        rounded corners, 
        fill=sbblue!50!white,
    },
    condition/.style={draw=none,
        diamond,
        fill=sbgreen!50!white,
        rounded corners,
    },
    measurement/.style={base,
        trapezium, 
        trapezium left angle = 65,
        trapezium right angle = 115,
        trapezium stretches,
        fill=sbred!50!white,
        rounded corners,
    },
    arr/.style={->,
        line width=1pt,
        color=black!90!white,
        -latex,
    },
    line/.style={line width=1pt,
        color=lightgray,
    },
    label/.style={pos=0.5,
        fill=white,
        inner sep=0.5ex,
    },
}
\begin{tikzpicture}[font=\small,thick]

\node (init) [box,
    ] {Init\\
        $\ket{\psi} \leftarrow \ket{\psi_0}$
    };
\node (while) [condition,
    below=5ex of init,
    ] {while
        $
        \neg \mathcal{R}
        $
    };
\node (freshancilla) [box,
    below=7ex of while,
    ] {Fresh ancilla\\
        $
        \ket{\psi}
        \otimes
        \ket{+} 
        $
    };
\node (entangling) [box,
    below=5ex of freshancilla,
    ] {Entangling\\
        $
        \ee^{-\ii C \otimes Y}
        \ket{\psi}
        \otimes
        \ket{+} 
        $
    };
\node (midmeasure) [measurement,
    below=5ex of entangling,
    ] {Measure ancilla
    };
\node (measure) [measurement,
    text width=6em,
    minimum width=5em,
    right=4em of while,
    ] {Measure main register
    };
\node (result) [box,
    fill=white,
    text width=6em,
    minimum width=5em,
    below=5ex of measure,
    ] {resulting bitstring $\mathbf{x}$
    };

\draw[arr] (init) -- (while);
\draw[arr] (while) -- (freshancilla) node[label]{Yes};
\draw[arr] (while) -- (measure) node[label]{No};
\draw[arr] (freshancilla) -- (entangling);
\draw[arr] (entangling) -- (midmeasure);
\draw[arr] (measure) -- (result);
\draw[arr, rounded corners] (midmeasure.south) 
    -- ++(0,-5ex) 
    -- ++(-7em, 0) 
    |- (while.west) 
    node[label, pos=0.25]{$b$}
    ;

\end{tikzpicture}

         }
    \end{center}
    \caption{
        Flow chart for Algorithm~\ref{alg:mdqo}. This algorithm may represent an instantiation of the blue box in Figure~\ref{fig:flow-general}.
        \label{fig:flow-mdqo}
    }
\end{figure}

The weak measurement operation is implemented via lines \ref{alg:line:mdqo_fresh_ancilla} through \ref{alg:line:mdqo_measurement} of Algorithm~\ref{alg:mdqo} (cf.~Figure~\ref{fig:flow-mdqo}). 
As we show below 
the cumulative effect of these operations depends only on 
the relative proportions of the measurement outcomes, not the order in which they occur. 
The return criterion $\mathcal{R}(\mathbf{b})$ for the while loop is based on the set of obtained measurement outcomes. 
When a high proportion of successes have occurred, the state has been favorably transformed and we measure the main register to obtain a candidate solution. 
Conversely, if relatively few successful outcomes have occurred, we can restart immediately. 
In this case we can measure the main register regardless which produces an additional candidate solution sample. 

\subsection{Algorithm Details}\label{sec:mdqo_algo_details}
We next explain and analyze the algorithm and its components in greater detail.
\paragraph*{Preliminaries:} Consider unconstrained optimization problems over bitstrings, 
which can be expressed~\cite{hadfield2021representation} as diagonal Hamiltonians in terms of Pauli $Z$ operators 
\begin{equation} \label{eq:costHam}
    H = 
    h_0 I + 
    \sum_{u=1}^n
    h_u
    Z_u 
    +
    \sum_{u<v} 
    h_{uv}
    Z_u Z_v + \dots \; ,
\end{equation}
derived from a pseudo-boolean cost function $h(x)$ we seek to optimize, such that $H\ket{x}=h(x)\ket{x}$ for all strings $x\in\{0,1\}^n$.
Without loss of generality consider the case of maximization. 
Constraint satisfaction problems such as MaxCut or Max$k$Sat are motivating examples. 
A wide variety of such problems are NP-hard to optimize, and as such we can only hope to efficiently find approximate solutions in the worst-case. 
A given solution $x$ achieves approximation ratio $h(x)/h^*$, where $h^*$ is the optimal cost value. An algorithm, quantum or classical, achieves approximation ratio $r$ if it is guaranteed to always output a solution $y$ with $h(y)/h^* \geq r$. For probabilistic approaches a common metric is the expected cost value (or expected approximation ratio), as we consider in our numerical examples to follow.\footnote{Generally speaking, algorithms that lack rigorous performance bounds but may perform well in practice are known as \textit{heuristics}, which include the algorithms of this paper and is common across many quantum approaches to optimization.} 

Given $H$, we will derive a new transformed 
cost Hamiltonian $C:=\mathcal{E}(H)$, which hence satisfies $C\ket{x}=c(x)\ket{x}$ for the transformed cost function $c= \mathcal{E} \circ h$. We next analyze the weak measurement step before explaining how the cost Hamiltonian transformation $\mathcal{E}$ is derived. 

\paragraph*{Weak measurements:} The core of the algorithm are 
the weak measurements applied to the main register probabilistically
\begin{align} \label{eq:ent-meas-operation}
    W_b 
    &:= 
    \bra{b} \ee^{-\ii C \otimes Y} \ket{+}
    \\
    &= 
    \sum_x 
    \underbrace{\frac{1}{\sqrt{2}}
    \left(
    \cos{c(x)} 
    + (-1)^{b \oplus 1}
    \sin{c(x)}
    \right)
    }_{=
        \begin{cases}
            \cos\left(c(x) + \frac{\pi}{4}\right) 
            &\text{for } b=0
            \\
            \sin\left(c(x) + \frac{\pi}{4}\right) 
            &\text{for } b=1
        \end{cases}
    }
    \ket{x}
    \bra{x} \end{align}
for possible outcomes $b \in \{0, 1\}$, 
As explained this operation is affected through appending, entangling and measuring an ancilla qubit in Algorithm~\ref{alg:mdqo}. 
Observe that $[W_0,W_1]=0$, which implies that the particular ordering of a sequence of weak measurement outcomes does not matter in terms of the resulting state, only the count of each outcome. Furthermore, as the terms of $C$ mutually commute, so do those of $C\otimes Y$, and hence we can assume the entangling gate and overall weak measurement are implemented exactly (i.e., Trotterization error is zero here, which is distinct from the approach of~\cite{mao2023measurement}). 
Therefore repeated applications of the combined entangling and measurement operations onto some initial state $\ket{\psi_0}$ result in
\begin{equation}
    \ket{\psi_{k_0, k_1}}
    \propto
    W_0^{k_0}
    W_1^{k_1}
    \ket{\psi_0}
    =
    \sum_x A_{k_0 k_1} (c(x)) \braket{x}{\psi_0} \ket{x},
    \label{eq:ent-meas-state}
\end{equation}
where $k_0,k_1$ denote the number of $0,1$ outcomes respectively, and 
\begin{equation} \label{eq:ampMod}
    A_{k_0, k_1}(c) 
    := 
    \cos^{k_0}\left(c + \frac{\pi}{4}\right)
    \sin^{k_1}\left(c + \frac{\pi}{4}\right)
\end{equation}
is the \emph{amplitude modulation function}.
The trigonometric form of 
this function suggests that
we can choose the rescaling $C=\mathcal{E}(H)$ (as to sufficiently bound its spectrum),
such that the $\sin^{k_1}\left(c + \frac{\pi}{4}\right)$ is monotonically increasing over the values of $c$. 
Hence, with each successful weak measurement step, the quantum amplitudes of higher cost states are amplified more that those of lower cost, as desired. 
We emphasize that 
$A_{k_0, k_1}$ only \textit{modulates} the state amplitude with respect to cost; the resulting amplitudes also critically depend on the support of the state they are applied to, i.e. the factors $\braket{x}{\psi_0}$ in Equation~\eqref{eq:ent-meas-state}. 
Said differently, for any sequence of measurement outcomes, the above analysis 
shows that the action of repeated weak measurements is diagonal with respect to the computational basis. 

For large $k_0+k_1$ the amplitude modulation becomes approximately Gaussian~\cite{mao2023measurement,ochoa2018simultaneous},
with peak position depending on the 
aggregated measurement outcomes 
given by
\begin{equation}
    P(k_0, k_1) := \frac{1}{2} \asin(\frac{k_1 - k_0}{k_0 + k_1}) \; .
    \label{eq:peak-position}
\end{equation}
This behavior is depicted in Figure~\ref{fig:amplitude-modulation-function-overview}. Observe that the peak position moves towards higher cost values for high fractions of desired outcomes $b=1$, and towards lower cost values for high fractions of undesired outcomes $b=0$. 
With increasing number of total measurements the peaks get narrower.
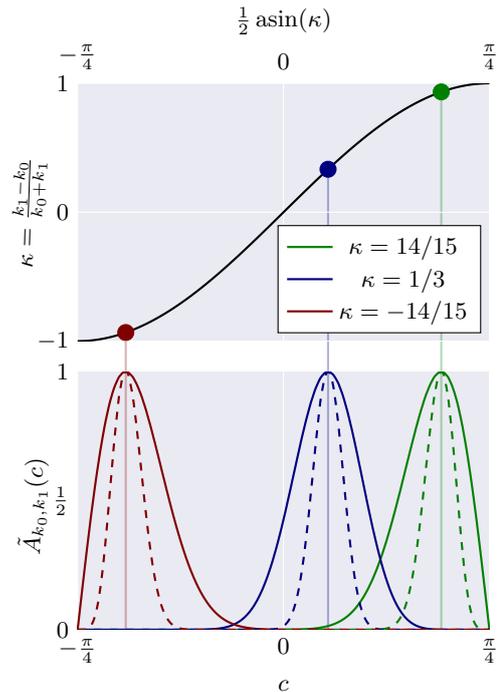
\begin{figure}[htb]
    \begin{center}
        \externalize{amplitude-modulation-function-overview}{\pgfplotsset{
    layers/my layer set/.define layer set={
        background,
        intermediate,
        main,
        foreground
    }{
background/.code={\pgfsetlayers{background}},
        intermediate/.code={\pgfsetlayers{intermediate}},
        main/.code={\pgfsetlayers{main}},
        foreground/.code={\pgfsetlayers{foreground}},
    },
    set layers=my layer set,
}
\begin{tikzpicture}
    \begin{groupplot}[
        group style={group size=1 by 2, 
vertical sep=3ex,
            ylabels at=edge left,
        },
trig format=rad,
         width=7cm,
         xmin=-pi/4,
         xmax=pi/4,
         xtick={-pi/4, 0.0, pi/4},
         xticklabels={$-\frac{\pi}{4}$, $0$, $\frac{\pi}{4}$},
         xlabel={$c$},
         xmajorgrids,
         tick style={draw=none}, axis line style={draw=none},
         grid style={color=white},
    ]
    \nextgroupplot[ylabel={$\kappa=\frac{k_1 - k_0}{k_0 + k_1}$},
xlabel={$\frac{1}{2} \asin(\kappa)$},
        y label style={rotate=-pi/6.28, anchor=center, align=center},
        yticklabels={$-1$, $0$, $1$},
        ytick={-1.0, 0.0, 1.0},
        ymajorgrids,
        samples=400,
        height=5cm,
        ymin=-1.0,
        ymax=1.0,
        axis x line=top,
    ]
    \begin{pgfonlayer}{background}
        \fill[sbbackground] (rel axis cs:0,0) rectangle (rel axis cs:1,1);
    \end{pgfonlayer}
    \addplot[name path=G,black!50!black,thick,domain={-pi/4:pi/4}]{sin(2*x)};

    \tikzmath{\t = 30;}
    \tikzmath{\kzero = 1;
        \kone = \t - \kzero;
        \x = asin((\kone-\kzero)/(\kzero + \kone))/2.0;
        \y = sin(2*\x);
        \xr = \x;
        \yr = \y;
    };
    \coordinate (topr) at (axis cs:\xr,\yr);
    \addplot[only marks, mark=*, mark size=3pt, green!50!black] coordinates {(\x, \y)};
    \tikzmath{\kzero = 10;
        \kone = \t - \kzero;
        \x = asin((\kone-\kzero)/(\kzero + \kone))/2.0;
        \y = sin(2*\x);
        \xm = \x;
        \ym = \y;
    };
    \coordinate (topm) at (axis cs:\xm,\ym);
    \addplot[only marks, mark=*, mark size=3pt, blue!50!black] coordinates {(\x, \y)};
    \tikzmath{\kzero = 29;
        \kone = \t - \kzero;
        \x = asin((\kone-\kzero)/(\kzero + \kone))/2.0;
        \y = sin(2*\x);
        \xl = \x;
        \yl = \y;
    };
    \coordinate (topl) at (axis cs:\xl,\yl);
    \addplot[only marks, mark=*, mark size=3pt, red!50!black] coordinates {(\x, \y)};

    \nextgroupplot[ylabel={$\tilde{A}_{k_0, k_1}(c)$},
        y label style={yshift=0.5em, rotate=-pi/6.28, anchor=center, align=center},
        yticklabels={$0$, $\frac{1}{2}$, $1$},
        ytick={0.0, 0.5, 1.0},
        height=5cm,
        samples=400,
legend entries={$\kappa=14/15$, $\kappa=1/3$, $\kappa=-14/15$},
        legend style={at={(0.98,1.15)},anchor=south east},
        ymin=0.0,
        ymax=1.0,
    ]
    \begin{pgfonlayer}{background}
        \fill[sbbackground] (rel axis cs:0,0) rectangle (rel axis cs:1,1);
    \end{pgfonlayer}

    \tikzmath{\t = 30;}
    \tikzmath{\kzero = 1;
        \kone = \t - \kzero;
        \m = asin((\kone-\kzero)/(\kzero + \kone))/2.0;
        \ml = \m;
    };
    \coordinate (bottomr) at (axis cs:\ml,0);
    \addplot[name path=G,green!50!black,thick,domain={-pi/4:pi/4}]{((cos(x + pi/4)/(cos(\m + pi/4.0)))^\kzero)*((sin(x + pi/4)/(sin(\m + pi/4)))^\kone)};
    \tikzmath{\kzero = 10;
        \kone = \t - \kzero;
        \m = asin((\kone-\kzero)/(\kzero + \kone))/2.0;
        \mm = \m;
    };
    \coordinate (bottomm) at (axis cs:\mm,0);
    \addplot[name path=G,blue!50!black,thick,domain={-pi/4:pi/4}]{((cos(x + pi/4)/(cos(\m + pi/4.0)))^\kzero)*((sin(x + pi/4)/(sin(\m + pi/4)))^\kone)};
    \tikzmath{\kzero = 29;
        \kone = \t - \kzero;
        \m = asin((\kone-\kzero)/(\kzero + \kone))/2.0;
        \mr = \m;
    };
    \coordinate (bottoml) at (axis cs:\mr,0);
    \addplot[name path=G,red!50!black,thick,domain={-pi/4:pi/4}]{((cos(x + pi/4)/(cos(\m + pi/4.0)))^\kzero)*((sin(x + pi/4)/(sin(\m + pi/4)))^\kone)};

\tikzmath{\t = 150;}
    \tikzmath{\kzero = 5;
        \kone = \t - \kzero;
        \m = asin((\kone-\kzero)/(\kzero + \kone))/2.0;
    };
    \addplot[name path=G,green!50!black,thick,dashed,domain={-pi/4:pi/4}]{((cos(x + pi/4)/(cos(\m + pi/4.0)))^\kzero)*((sin(x + pi/4)/(sin(\m + pi/4)))^\kone)};
    \tikzmath{\kzero = 50;
        \kone = \t - \kzero;
        \m = asin((\kone-\kzero)/(\kzero + \kone))/2.0;
    };
    \addplot[name path=G,blue!50!black,thick,dashed,domain={-pi/4:pi/4}]{((cos(x + pi/4)/(cos(\m + pi/4.0)))^\kzero)*((sin(x + pi/4)/(sin(\m + pi/4)))^\kone)};
    \tikzmath{\kzero = 145;
        \kone = \t - \kzero;
        \m = asin((\kone-\kzero)/(\kzero + \kone))/2.0;
    };
    \addplot[name path=G,red!50!black,thick,dashed,domain={-pi/4:pi/4}]{((cos(x + pi/4)/(cos(\m + pi/4.0)))^\kzero)*((sin(x + pi/4)/(sin(\m + pi/4)))^\kone)};
    \end{groupplot}
    \begin{pgfonlayer}{intermediate}
    \draw[green!50!black, opacity=0.3, thick] (topr) -- (bottomr);
    \draw[blue!50!black, opacity=0.3, thick] (topm) -- (bottomm);
    \draw[red!50!black, opacity=0.3, thick] (topl) -- (bottoml);
    \end{pgfonlayer}
\end{tikzpicture}

         }
    \end{center}
    \caption{\label{fig:amplitude-modulation-function-overview}
        Peak position (top) of the normalized amplitude modulation function $\tilde{A}_{k_0, k_1}(c)=\frac{A_{k_0, k_1}(c)}{\max_c A_{k_0, k_1}(c)}$ (bottom) for different number of measurement outcomes. 
        The total number of measurements $K=k_0 + k_1$ is $K=30$ (solid lines) and $K=150$ (dashed lines).
    }
\end{figure}

We next explain how to rescale the cost Hamiltonian to lie in a suitable spectral range as described, and furthermore how to do this in a way as to guarantee 
the success probability of each weak measurement is bounded away from~$1/2$.   

\paragraph*{Transformed Hamiltonian:} \label{sec:transformedHam}
For a given cost Hamiltonian $H$, first consider the simple case of creating a state 
$\ket{\psi}\propto H\ket{\psi_0}$. 
Measuring $\ket{\psi}$ then yields a sample $x$ with probability proportional to $h(x)^2$. Hence if the cost function ranges over both positive and negative numbers, this will have the effect of amplifying the probability of both desirable high-cost as well as undesirable low-cost states. 

\begin{figure}[htb]
    \begin{center}
        \externalize{amplitude-modulation-function-squared-single-measurement}{\begin{tikzpicture}
\begin{axis}[trig format=rad,
        width=7cm,
        height=4cm,
        xmin=-pi/4,
        xmax=pi/4,
        xtick={-pi/4, 0.0, pi/4},
        xticklabels={$-\frac{\pi}{4}$, $0$, $\frac{\pi}{4}$},
        xmajorgrids,
        xlabel={$c=\mathcal{E}(h)$},
        ymin=0.0,
        ymax=1.0,
        yticklabels={$0$, $\frac{1}{2}$, $1$},
        ytick={0.0, 0.5, 1.0},
        ymajorgrids,
tick style={draw=none}, axis line style={draw=none},
        grid style={color=white},
    ]
    \begin{pgfonlayer}{background}
        \fill[sbbackground] (rel axis cs:0,0) rectangle (rel axis cs:1,1);
        \fill[sbbackground!50!white] (rel axis cs:0,0) rectangle (rel axis cs:0.5,1);
    \end{pgfonlayer}

    \addplot[name path=G,red!50!black,thick,domain={0:pi/4}]{cos(x + pi/4)^2};
    \addplot[name path=G,green!50!black,thick,domain={0:pi/4}]{sin(x + pi/4)^2};
    \addplot[name path=G,red!50!black,thick,dashed,domain={-pi/4:0}]{cos(x + pi/4)^2};
    \addplot[name path=G,green!50!black,thick,dashed,domain={-pi/4:0}]{sin(x + pi/4)^2};
    \node [color=red!50!black,
        anchor=north west,
        inner sep=1ex] at (axis cs:-pi/4, 0.8) {\footnotesize $\cos^2(c + \frac{\pi}{4})$};
    \node [color=green!50!black,
        anchor=north east,
        inner sep=1ex] at (axis cs:pi/4, 0.8) {\footnotesize $\sin^2(c + \frac{\pi}{4})$};
\end{axis}
\end{tikzpicture}

         }
    \end{center}
    \caption{\label{fig:amplitude-modulation-function-single-measurement}
        Squared amplitude modulation functions for a single weak measurement operation (cf.~Eq.~\eqref{eq:ent-meas-operation}),
        in the domain of the rescaled spectrum (solid lines) and beyond (dashed lines).
    }
\end{figure}
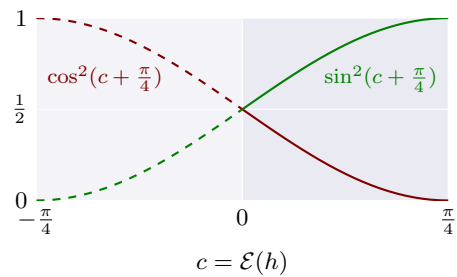

To restore the desired behaviour we introduce a scaling function $\mathcal{E}$ to shift the spectrum of the transformed Hamiltonian $C=\mathcal{E}(H)$ to be 
bounded by constants 
$a \leq C \leq b$ 
For the weak measurement step described above we will use 
\begin{equation} \label{eq:Hambounds}
    0 \leq C \leq \pi/4.
\end{equation}
As one can see in Figure~\ref{fig:amplitude-modulation-function-single-measurement}, 
this ensures 
greater magnification of measurement probability for higher cost states across the whole rescaled spectrum, in the case of a successful weak measurement step.
Furthermore, we seek to saturate the bounds in the sense that $\min_x c(x) \simeq 0$ and $\max_x c(x) \simeq \pi/4$. 
Here the upperbound is chosen so 
that $\pi/4+\max_xc(x)$ is as close as possible to $\pi/2$ as to maximize the $\sin(\cdot)$ function term in Eq. \eqref{eq:ampMod} above. As we explain below the lower bound is chosen to favorably shift the success probability. 

In general, we may consider $C = \mathcal{E}(H)$ to be any function of $H$, such as a suitable polynomial, that ensures Eq.~\eqref{eq:Hambounds} is satisfied. 
Nevertheless in the remainder of this work we focus on the linear case 
\begin{equation}
    C := \mathcal{E}(H) = \varepsilon(\alpha I + H)
    \label{eq:scaling_function_linear}
\end{equation}
for suitable constants $\varepsilon,\alpha$ that are derived from the problem instance~$H$. 
Clearly, obtaining the tightest possible upper bound on $c(x)$ is as hard as the problem itself. 
Suppose we know bounds on the spectrum of the input cost Hamiltonian $H$ such that 
\begin{equation}
    -s\leq H \leq t.
    \label{eq:scaling_bounds}
\end{equation}
These can always be obtained from the coefficients of Eq.~\eqref{eq:costHam}, 
which we assume we are given. 
In some cases tighter 
bounds can be derived from the input problem instance. 
Then we can use these values to select $\alpha,\varepsilon$ as to satisfy Eq.~\ref{eq:Hambounds}. Indeed, setting $\alpha=s$ suffices to ensure $C\geq 0$. Then setting $$\varepsilon=\frac{\pi}{4(s+t)}$$ 
implies $C\leq \pi/4$ as desired. 
In practice one may have limited information in terms of tight bounds to $H$ and so could search, i.e.~repeat the algorithm, over different values of $B$ (which determines $\alpha,\varepsilon$). 
For example, for MaxCut one could use a set of values $B=|E|,|E|-1, \dots$, 
or for larger ranges binary search could be employed. 

We emphasize that while in many quantum algorithms the effect of explicitly including Identity terms in the cost Hamiltonian is trivial, in our case their inclusion effects the entangling operation of Algorithm~\ref{alg:mdqo} (albeit, at the overall expense of a single-qubit Pauli rotation gate applied to the ancilla qubit). 
\paragraph*{Example: Max-Cut problem.}
Given a graph $G=(V, E)$ with $n=|E|$ nodes and $m=|E|$ edges.
The Hamiltonian we seek to maximize is
\begin{equation}
    H = 
    \frac{m}{2} I-\frac12 \sum_{(uv) \in E} 
    Z_u Z_v
    \label{eq:max_cut_hamiltonian}
\end{equation}
which encodes the number of cut edges $h(x)=\sum_{(ij)\in E} x_i \oplus x_j$. 
Retaining the identity term $\frac{m}{2} I$ in $H$ ensures $H\geq 0$. 
Only bipartite graphs have $h^*=m$, which is an efficiently checkable property. Hence we could choose $\varepsilon = \frac{\pi}{4m}$ and $\alpha=0$ in \eqref{eq:scaling_function_linear} as to guarantee 
$C\in [0, \frac{\pi}{4})$. 

\begin{rem}\label{rem:bound-support}
    Consider a weak measurement step applied to a state $\ket{\psi}$. 
    If we know bounds on the support of $\ket{\psi}$, then better constants can be derived by considering the effective cost Hamiltonian on this subspace. For example, if promised that $\ket{\psi}$ is only supported on states of positive cost of $H$, no shift factor $\alpha$ is needed, and the rescaling factor $\varepsilon$ can be selected larger than otherwise.
    Indeed suppose we know constants $\tilde{s},\tilde{t}$ such that $s \leq \tilde{s} \leq H|_{\text{support}(\ket{\psi})} \leq \tilde{t} \leq t$. Then by the same argument as above it suffices to use $\alpha = \tilde{s}$ and $\varepsilon =\frac{\pi}{4(\tilde{s}+\tilde{t})}$. 
    
    At the same time we reemphasize that $\ket{\psi}$ necessarily must have support on target states; any number of steps of Algorithm 1 cannot produce states outside of the span of the initial support. As this support may be unknown 
    we propose 
    the incorporation of additional operators 
        to alleviate this issue 
    in Algorithm~\ref{alg:fcmdqo} of Section~\ref{sec:fcmdqo}.
\end{rem}

\paragraph*{Return Criterion:}
The return criterion 
specifies when we decide to measure the main register to obtain a candidate solution.
Here, by design, we incorporate as much flexibility as possible as to accommodate wider strategies appropriate to emerging quantum hardware with their variety of space, depth, and noise tradeoffs, as well as potential speed and timing limitations applicable to mid-circuit measurements. 
To this end we suggest and consider several strategies. 

The first criteria is to return when we believe we achieved a distribution concentrated on high quality solutions.
We can utilize the cumulative weak-measurement outcomes 
that determine the peak position of the amplitude modulation function~\eqref{eq:peak-position},
and trigger the return (measure the main register) if the peak position
surpasses
a specified threshold value $T$, where $0 \leq \mathcal{E}(T) \leq \frac{\pi}{4}$. 
The value of $T$ can be set based on problem properties, such as bounds to the optimal cost, or adaptively using the cost of the best cost and distribution of solutions already found. 
In this case the return criterion would read
\begin{equation}
    \mathcal{R}(\mathbf{b}) 
    = 
   \big( P(k_0, k_1) \geq \mathcal{E}(T)\big)
    \label{eq:return_criterion_default}
\end{equation}
which effectively defines a target region in between the threshold and maximum cost value $[T, h^*]$. 

An alternative criterion which requires even simpler classical feed-forward processing, would be to 
return if a minimum difference in desirable and undesirable outcomes $L$ is reached
\begin{equation}
    \mathcal{R}(\mathbf{b}) 
    = 
    \big(k_1 - k_0 \geq L \big) \; .
    \label{eq:return_criterion_simple}
\end{equation}

For the inner loop, an important criterion which does not require classical feed forward processing as above,
is to terminate
when we reach some ceiling of $K_\mathrm{T}$ steps,
\begin{equation}
\mathcal{R}(\mathbf{b}) 
    = 
    \big( k_1 + k_0 \geq K_\mathrm{T} \big) \; .
    \label{eq:return_criterion_ceiling}
\end{equation}
The value $K_\mathrm{T}$ may relate to hardware constraints such as coherence time, and is especially relevant to near-term devices.

Finally, if a partial sequence of measurement outcomes indicates we are unlikely to reach a satisfactory distribution we can save considerable resources by returning early. In particular if a substantial number of bad outcomes have occurred the procedure can be terminated and restarted well before $K_\mathrm{T}$ trials are reached. Hence we additionally 
propose a so-called reset condition triggered if ever fixed difference $R$ of failures to successes is reached,
\begin{equation}
    \mathcal{R}_\mathrm{R}(\mathbf{b}) 
    = 
\big( k_0 - k_1 \geq R \big) \; .
    \label{eq:reset_criterion}
\end{equation}
We elaborate on this condition and its implications in Section~\ref{sec:multistep}. 

In practice one may incorporate all or a subset of the above conditions. Additionally, 
a so-called burn-in may be utilized, where the reset or other return conditions above become active only after a fixed number of initial weak measurement steps.

\paragraph*{Adaptive algorithm variants:}
The algorithm parameters, i.e., those of the scaling function $\mathcal{E}$ and the return criterion $\mathcal{R}$, can adaptively changed during execution, toward sampling the best solution possible. 
Parameters could be set as to enforce targeting of solutions better than the best already found, or tuned based on the performance of all past runs.
Adaptive parameter setting may be applied 
within the inner loop of~Figure~\ref{fig:flow-general} informed by mid-circuit measurement outcomes 
(which allow for the calculation of the peak position of the amplitude modulation function~\eqref{eq:peak-position}), 
or in the outer loop informed by sampled bitstrings as a result of main register measurements. 
For example the parameter $L$ in the return criterion~\eqref{eq:return_criterion_simple} in the outer loop
can be progressively increased to target distributions likely to improve on the best solution already found. 

Similarly, the threshold value $T$ in the return criterion~\eqref{eq:return_criterion_default} 
may be adaptively increased, 
effectively narrowing down the target region.
Starting with a low value (e.g., corresponding to solutions obtained from random guessing or oter classical algorithms), 
the threshold value is subsequently increased by setting it to the best know cost value so far (outer loop), or to a value informed by the highest peak-position so far (inner loop).
Analogue considerations apply for 
the lower bound in \eqref{eq:scaling_bounds}.

\subsection{Resource estimation} \label{sec:resource_est}
Here we explain that if we can implement the $n$-qubit operator 
$U=\exp(-\ii t C), t\in \mathds{R}$
efficiently, then we can implement the $(n+1)$-qubit operator $U=\exp(-\ii t C \otimes Y)$ with small additional overhead.
Recall we assume linear rescaling $C= \varepsilon(\alpha I + H)$.
A key property of combinatorial optimization our algorithm takes advantage of is the fact that $H$ is diagonal, and can be decomposed into a sum of $m$ mutually commuting terms $H=\sum_{j=1}^m H_j$ as in Eq.~\ref{eq:costHam}. 
Commutation implies that $\exp(-\ii t H)=\prod_j \exp(-\ii t H_j)$, with efficient quantum circuits exist for implementing each $\exp(-\ii t H_j)$. 
Hence if $m=\textrm{poly}(n)$ then $\exp(-\ii t H)$ can be implemented efficiently.
The same ideas apply to implementing the operator $\exp(-\ii C \otimes Y)=\exp(-\ii \varepsilon (\alpha I+H)\otimes Y)$ which can be similarly decomposed using $\exp(-\ii \varepsilon H\otimes Y)=\prod_j \exp(-\ii \varepsilon H_j \otimes Y)$.  
As indicated in Figure~\ref{fig:resources}, 
using the approach and basic gate set of \cite{barenco1995elementary}, $U=\exp(-\ii t C \otimes Y)$ can be implemented using $O(mk)$ CNOT gates and $O(m)$ single-qubit rotation gates, when the cost Hamiltonian $C$ contains $m$ terms of degree at most $k$. 
We emphasize the inclusion of the linear shift $\alpha$ requires only one additional single-qubit Pauli rotation gate, yet has an important effect on the success probability and performance, as discussed. 
In terms of classical resources, for Algorithm $1$ the classical coprocessing of the mid-circuit measurement outcomes is only required to compute 
whether the threshold conditions are satisfied which can also be accomplished with minimal overhead. 
\begin{figure}[htb]
    \begin{center}
        \begin{tikzpicture}
            \node (qaoa) []{\externalize{circuit_resources_qaoa}{\begin{quantikz}[column sep=0.5em, row sep={5ex, between origins}]
    \lstick[1]{$i$}
    & 
    & \ctrl{1}
    &
    & \ctrl{1}
    & 
    \\
    \lstick[1]{$j$}
    &
    & \targ{}
    & \gate{R_Z}
    & \targ{}
    & 
\end{quantikz}
                 }
            };
            \node (mdqo) [below=0.5em of qaoa,
                ]{\externalize{circuit_resources_mdqo}{\begin{quantikz}[column sep=0.4em, row sep={5ex, between origins}]
    \lstick{$\ket{+}$}
    & \gate{R_X(\frac{\pi}{2})}
    & \ctrl{1}
    &
    &
    &
    & \ctrl{1}
    & \gate{R_X(\frac{-\pi}{2})}
    & \ \ldots \
    & \meter{}
    \\
    \lstick[1]{$i$}
    & 
    & \targ{}
    & \ctrl{1}
    &
    & \ctrl{1}
    & \targ{}
    &
    & \ \ldots \
    &
    \\
    \lstick[1]{$j$}
    & 
    &
    & \targ{}
    & \gate{R_Z}
    & \targ{}
    &
    & 
    & \ \ldots \
    &
\end{quantikz}
                 }
            };
        \end{tikzpicture}
    \end{center}
    \caption{Resource overhead from implementing 
        a partial step in Algorithm~\ref{alg:mdqo} (bottom) 
        compared to a
        a single QAOA phase-separation layer (top), both
        exemplified for single quadratic term $Z_i Z_j$ in a cost Hamiltonian. }
    \label{fig:resources}.
\end{figure}
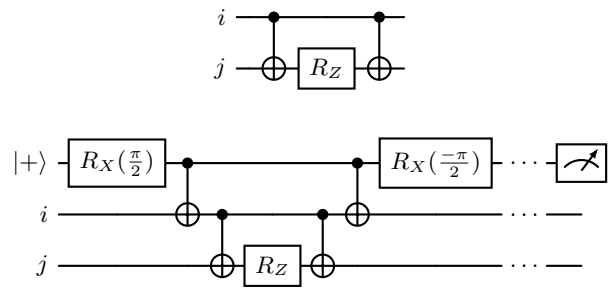

\subsection{Numerical Verification}\label{sec:numerics_mdqo}
Here we provide numerical evidence to support the claims made thus far as well as the analysis to follow. 
As many of the properties of our algorithm are fairly generic, we focus on a single fixed problem instance in order to elucidate these qualities.  
\begin{figure}[htb]
    \begin{center}
        \begin{tikzpicture}\node (plot) []{\externalize{results_mdqo}{\pgfplotstableread[col sep=comma]{./figures/tikz/data/results_mdqo_exp_value_best.csv}{\expvalbest}
\pgfplotstableread[col sep=comma]{./figures/tikz/data/results_mdqo_exp_value_worst.csv}{\expvalworst}
\pgfplotstableread[col sep=comma]{./figures/tikz/data/results_mdqo_success_prob_best.csv}{\successprobbest}
\pgfplotstableread[col sep=comma]{./figures/tikz/data/results_mdqo_success_prob_worst.csv}{\successprobworst}

\pgfplotstablegetcolsof{\expvalbest}
\pgfmathsetmacro{\numcols}{\pgfplotsretval}

\pgfmathsetmacro{\lastycol}{\numcols - 1}

\pgfplotscreateplotcyclelist{color list}{
  {sbblue, solid, thick},
  {sbred, solid, thick},
  {sbgreen, solid, thick},
  {sborange, solid, thick},
  {sbpurple, solid, thick},
{sbblue, dashed, thick},
  {sbred, dashed, thick},
  {sbgreen, dashed, thick},
  {sborange, dashed, thick},
  {sbpurple, dashed, thick},
}
\begin{tikzpicture}[trim axis left]
    \begin{groupplot}[
        group style={group size=1 by 2, 
vertical sep=8ex,
            ylabels at=edge left,
        },
    ]
    \nextgroupplot[width=8cm,
        height=5cm,
tick style={draw=none}, ymax=5.1,
        ytick={0, 1, 2, 3, 4, 5},
        enlarge x limits=0.05,
        enlarge y limits=0.05,
        xlabel={$L=k_1 - k_0$},
        ylabel={$\langle H \rangle$},
        legend style={draw=none,
opacity=0.8,
            inner sep=1pt,
            legend columns=3,
            column sep=0.1ex,
            at={(0.99,0.05)},
            anchor=south east,
        },
        legend image post style={xscale=0.5}, grid=both,
        cycle list name=color list,
        axis background/.style={fill=sbbackground},
        axis line style={draw=none},
        grid style={color=white},
    ]

    \addlegendimage{empty legend}
    \addlegendentry[]{$k_0$}
\foreach \i in {1,...,\lastycol} {

\pgfplotstablegetcolumnnamebyindex{\i}\of{\expvalbest}\to\pgfplotsretval
\addplot table [
            x index=0,
            y index=\i,
        ] {\expvalbest};
        \addlegendentryexpanded{\pgfplotsretval}
    }
    \foreach \i in {1,...,\lastycol} {

\pgfplotstablegetcolumnnamebyindex{\i}\of{\expvalworst}\to\pgfplotsretval
\addplot table [
            x index=0,
            y index=\i,
            dashed,
        ] {\expvalworst};
    }
    \addplot [only marks, 
              mark size=1.5pt, 
              mark=*,
              color=sbpurple, 
              ] coordinates { (110,2) };
    \nextgroupplot[width=8cm,
        height=5cm,
tick style={draw=none}, enlarge x limits=0.05,
        enlarge y limits=0.05,
        xlabel={$L=k_1 - k_0$},
        ylabel={$p_1$},
grid=both,
        cycle list name=color list,
        axis background/.style={fill=sbbackground},
        axis line style={draw=none},
        grid style={color=white},
    ]

\foreach \i in {1,...,\lastycol} {

\pgfplotstablegetcolumnnamebyindex{\i}\of{\successprobbest}\to\pgfplotsretval
\addplot table [
            x index=0,
            y index=\i,
        ] {\successprobbest};
}
    \foreach \i in {1,...,\lastycol} {

\pgfplotstablegetcolumnnamebyindex{\i}\of{\successprobworst}\to\pgfplotsretval
\addplot table [
            x index=0,
            y index=\i,
            dashed,
        ] {\successprobworst};
    }
    \addplot [only marks, 
              mark size=1.5pt, 
              mark=*,
              color=sbpurple, 
              ] coordinates { (110,0.794) };
    \end{groupplot}
\end{tikzpicture}

                 }
            };
            \node (graph) [right=-0.5em of plot.south east,
                yshift=8ex,
                anchor=south east,
                scale=0.8,
                ]{\externalize{results_mdqo_graph}{
\tikzset{%
    vertex/.style={%
        circle,
        draw=none,
        minimum width=1ex,
    },
    edge/.style={%
        ultra thick,
        sbgray,
    },
}
\def\radius{5ex}
\begin{tikzpicture}
    \node (1) [vertex, fill=sbgreen] at ({360/5*(1.25)}:\radius) {};
    \node (2) [vertex, fill=sbred] at ({360/5*(2.25)}:\radius) {};
    \node (3) [vertex, fill=sbred] at ({360/5*(3.25)}:\radius) {};
    \node (4) [vertex, fill=sbgreen] at ({360/5*(4.25)}:\radius) {};
    \node (5) [vertex, fill=sbgreen] at ({360/5*(5.25)}:\radius) {};

    \draw[edge, sbred] (1) -- (2);
    \draw[edge, sbgreen] (2) -- (3);
    \draw[edge, sbred] (3) -- (4);
    \draw[edge, sbred] (1) -- (3);
    \draw[edge, sbred] (2) -- (4);
    \draw[edge, sbred] (2) -- (5);

\end{tikzpicture}

                     }
            };
{};
        \end{tikzpicture}
    \end{center}
    \caption{\label{fig:results_mdqo}
Expectation value of $\langle H \rangle$ (i.e. number of cuts)  of the MaxCut Hamiltonian~\eqref{eq:max_cut_hamiltonian} for the displayed graph with respect to the state in \eqref{eq:ent-meas-state} (top) and probability for the desired outcome $p_1$ (bottom) against the difference in desired outcomes and undesired outcomes $L=k_1-k_0$ for various fixed numbers of undesired outcomes $k_0$. 
        The solid lines and dashed lines indicate the tight bound $B=h^*$ and loose bound $B=m$, respectively, used to determine the cost Hamiltonian rescaling parameters (see Section~\ref{sec:effectsConsts}).
        The markers indicate the starting point for simulations below (cf.~Figure~\ref{fig:results_fcmdqo}).
        }
\end{figure}
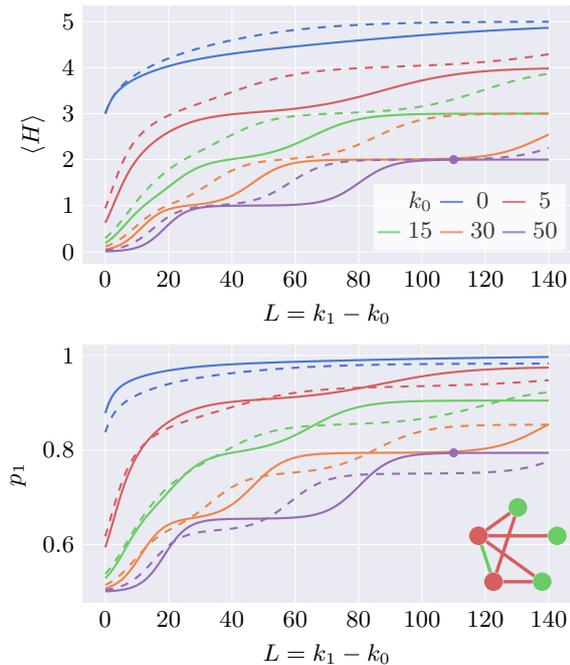

In Figure~\ref{fig:results_mdqo} the results of Algorithm~\ref{alg:mdqo} for a simple MaxCut instance are displayed.
As expected, the cost expectation value $\langle H \rangle$ increases monotonically with the number of desired outcomes $k_1$.
Plateau-like features are observed at eigenvalues of the cost Hamiltonian, where the rate of change slows down. 
In the example problem these are the numbers $\{0,1, \dots, 5\}$.
This is a result of the peak of the amplitude modulation function (cf.~Figure~\ref{fig:amplitude-modulation-function-overview}) intersecting with the eigenvalues according to Equation~\eqref{eq:ent-meas-operation}.
The performance decreases, as expected, with the number of undesired outcomes.

Recall the discussion of how to suitably bound and transform the cost Hamiltonian in Section~\ref{sec:transformedHam}. 
Here we consider two cases of using the ideal value ($h^*=5$), or using the bound ($B=m=6>h^*$). 
Remarkably, we observe that using the looser bound results in even better performance than the tight bound in terms of cost expectation value. On the other hand the single-step success probability is seen to closely track that of the ideal case, but not with strict improvement nor detriment. 
Hence these results 
demonstrate the effect of not knowing the optimal bound to the cost function is 
not the roadblock one might expect it to be, and surprising can even be sometimes beneficial. 
We 
discuss and analyze this effect further 
in the next subsection. 
Furthermore, 
we suggest the possibility or adaptively searching over different values of the cost bound $B$. 
Indeed, one could view $\varepsilon=\varepsilon(B)$ as a variational parameter to be further optimized.

\subsection{Algorithm Analysis: Single Step}
Consider a single weak measurement step 
(i.e., an iteration of the while loop of Algorithm 1) with state $\ket{\psi}$ in the main register; this could be the fiducial initial state, or the state resulting from some number of previous steps of the algorithm, or as we discuss below in Section~\ref{sec:postproc} it could be the state resulting from an optimized parameterized quantum circuit ansatz that one seeks to improve upon. 

\paragraph*{Success probability:}
For maximization problems, measuring a $1$ in the ancilla register corresponds to a "success". 
For a single weak measurement step the probability $p_1=p_1(\ket{\psi})$ of measuring $1$ is
\begin{equation} \label{eq:succProb}
    p_1 
    = \bra{\psi} W_1^\dagger W_1 \ket{\psi}
    = \frac{1}{2} + \frac{1}{2} \bra{\psi}  \sin 2 
    C \ket{\psi} 
\end{equation}
and similarly $p_0 = \frac{1}{2} - \frac{1}{2} \bra{\psi}  \sin 2 C \ket{\psi}$.

Observe that if we have 
shifted and rescaled 
the cost Hamiltonian such that
 $0\leq C \leq \pi/4$, then we are guaranteed that $p_1 \geq 1/2$. Moreover with the mild assumption that $\ket{\psi}$ is not a $0$ eigenstate of $C$ this inequality is strict, 
$$p_1 > \frac12 > p_0$$
as desired. 
Applying the identity $\sin(x)\geq \frac2\pi x$ on the interval $[0,\pi/2]$ gives the 
lower bound 
\begin{equation} \label{eq:probLowerBound}
    p_1 \geq \frac12 + \frac2{\pi} \bra{\psi}
    C \ket{\psi}  . 
\end{equation}

We next prove the desirable property that the success probability increases monotonically with each subsequent success.
\begin{prop} Consider a state $\ket{\psi}$ to which a successful weak measurement step of Algorithm 1 is applied, resulting in the normalized state $\ket{\phi}$ Then the probability $p_1$ of a subsequent 
    weak measurement being successful 
    satisfies
    \begin{equation} \label{eq:probIncrease}
        p_1(\ket{\phi}) \, \geq \, p_1(\ket{\psi}).
    \end{equation}
\end{prop}
\begin{proof}
    Let $S:=\sin(\pi/4+C)$. From Eq.~\eqref{eq:ent-meas-operation} above we have $p_1(\ket{\psi})=\langle S^2\rangle_\psi$ and 
    $p_1(\ket{\phi})=\langle S^4\rangle_\psi/p_1(\ket{\psi})$. As $S$ is monotonically increasing with respect to the spectrum of $C$ for any $\ket{\psi}$ we have $\langle C^4 \rangle \geq \langle C^2 \rangle^2$ which implies $p_1(\ket{\phi})\geq p_1(\ket{\psi})$ as claimed. 
\end{proof}

\begin{rem}\label{rem:positive-feedback}
    Equations \eqref{eq:succProb} and~\eqref{eq:probIncrease} 
reveal the positive feedback aspect of our approach; after a successful outcome, $\langle H\rangle$ typically improves, and so the probability of success
for the next weak measurement improves as well. 
This effect is clearly demonstrated in the bottom plot of Figure~\ref{fig:results_mdqo}. 
\end{rem}
    
\paragraph*{Effect on cost expectation:}
Here we show 
how under mild assumptions the expected cost always improves with each successful weak measurement step. 

\begin{prop} \label{lem:costIncrease}
    Consider a given cost Hamiltonian $H$ of Eq.~\eqref{eq:costHam} we seek to maximize, and a suitable rescaling $C=\varepsilon(\alpha I + H)$ as in Eq.~\eqref{eq:Hambounds}.
Suppose we are given a normalized quantum state $\ket{\psi}$ to which we apply a single weak measurement step of Algorithm $1$, which succeeds and results in new state $\ket{\psi'}$. Then if $\ket{\psi}$ is not an eigenstate of $H$ we have 
    $$ \langle H \rangle_{\psi'} >  \langle H \rangle_{\psi}$$
    as desired; else $\langle H \rangle_{\psi'} =  \langle H \rangle_{\psi}$.
\end{prop}
\begin{proof}
Suppose we obtain the weak measurement outcome $1$. Then the normalized state becomes
$$\ket{\psi'}=\frac{\sin(\frac{\pi}4I+C)\ket{\psi}}{\sqrt{p_1}}$$
and so
using $[H,C]=0$ along with standard trigonometric identities (see App.~\ref{app:trig}) we have for the cost expectation
\begin{align} 
    \langle H\rangle_{\psi'}
    &=\frac{\bra{\psi}H\sin^2(\frac{\pi}4I+C)\ket{\psi}}{p_1}
    \\
    &= \frac{ \langle H\rangle_{\psi} +\langle H\sin(2 C )\rangle_{\psi}}{1 +\langle  \sin 2 C \rangle_{\psi}}   .
\label{eq:changeExpecC}
\end{align}
Subtracting $\langle H\rangle_{\psi}$ from each side give the change in cost expectation to be
\begin{align} 
  \langle H\rangle_{\psi'} - \langle H\rangle_{\psi} 
    &= \frac{\langle H\sin(2 C )\rangle_{\psi}-\langle H\rangle_{\psi} \langle \sin(2 C )\rangle_{\psi}}
    {2p_1}.
    \label{eq:changeExpecC2}
\end{align}
From the assumption $C\geq 0$, and 
using that $x,\sin x$ are both non-negative and monotonically increasing on $[0,\pi/2]$, we have the inequality 
$$\langle  H\sin(2 C )\rangle_\psi \geq \; \langle H\rangle_\psi \cdot\langle \sin(2 C )\rangle_\psi$$ 
which plugging into \eqref{eq:changeExpecC} gives $\langle H\rangle_{\psi'}\geq \langle H\rangle_{\psi}$. 

It remains to show $\langle H\rangle_{\psi'}  > \langle H\rangle_{\psi}$ whenever $\ket{\psi}$ contains at least two components of different 
costs.
For the case that $\ket{\psi}$ is an eigenstate of $H$, i.e., $H\ket{\psi}=\lambda\ket{\psi}$ for some $\lambda$,
observe that irrespective of whether the weak measurement succeeds or fails, either applied operator acts as a constant and hence $\langle H \rangle_{\psi'} =  \langle H \rangle_{\psi}$. 
Alternatively, now suppose $\ket{\psi}$ contains components of at least two distinct cost eigenspaces. 
As seen from Eq.~\eqref{eq:ent-meas-operation}, 
the amplitude of each supported computational basis state $\ket{x}$ in multiplied by a factor proportional to $\sin(\pi/4+c(x))$. 
As the rescaling parameters are chosen such that 
$\sin(\pi/4+C)$ is monotonically increasing with respect to the spectrum of $C$, higher cost states are magnified more than lower cost ones, which gives the claim.
\end{proof}

\begin{rem}\label{rem:tradeoff}
    Observe the inverse tradeoff between success probability and improvement to the cost expectation evident in Eq.~\eqref{eq:changeExpecC},  which each depend on the particular state $\ket{\psi}$.  In practice one may be willing to tolerate slightly lower probability of success in  exchange for greater improvement when success occurs. 
\end{rem} 

\begin{rem}\label{rem:support-not-eigenstate}
In practice we may assume $\ket{\psi_0}$ is not an eigenstate of the cost Hamiltonian $H$. If it was, it would be easily detected, as computational basis measurements will return only strings of the same cost. In that case, the input ansatz $\ket{\psi_0}$ may be altered, 
or a suitable scrambling operator (see Section~\ref{sec:fcmdqo}) may be applied to  $\ket{\psi_0}$ to yield a new initial state without this property. 
\end{rem}

We next 
elaborate on how the choice of algorithm parameters further affect these quantities. 

\paragraph*{Effects of choice of problem rescaling:} 
\label{sec:effectsConsts}

We next elaborate on the consequences of using a suitable bound in place of the unknown optimal cost for our algorithm, which results in a smaller than necessary value of $\varepsilon$. 
Surprisingly, we show that the effect of this on the cost expectation value can be beneficial as opposed to detrimental. Furthermore, even with very loose bounds we retain the important property that the per step success probability remains at least~$1/2$. Here assume we know a value $B>h^*$ such that $0\leq H\leq B$. The results of this section easily extend to the case where both upper and lower bounds to the spectrum of $H$ are necessary.

In particular we explain that for most cost functions and input states $\ket{\psi}$, the normalized cost expectation $\langle H\rangle_{\psi'}$ is maximized for some 
$\varepsilon \in (0,\varepsilon(h^*))=(0,\frac{\pi}{4h^*})$.  
(By \textit{most} here we mean $h(x),\ket{\psi}$ such that $h(x)\cos(2\varepsilon h(x))$ is nonincreasing over the spectrum of $h$;
we give precise conditions in Appendix~\ref{app:analysis}.)
The main idea is illustrated in Figure~\ref{fig:cost_expectation_epsilon_sketch}. 
Let $\langle H\rangle_{\psi'} (\varepsilon)$ denote the (normalized) cost expectation as a function of the parameter $\varepsilon=\varepsilon(B)=\frac{\pi}{4B}$. In Appendix~\ref{app:analysis} we show 
\begin{equation} \label{eq:dHde}
   \frac{d}{d\varepsilon}\langle H\rangle_{\psi'} (\varepsilon)|_{\varepsilon=\varepsilon(h^*)} \; \leq \; 0 \; 
   \leq \; \frac{d}{d\varepsilon}\langle H\rangle_{\psi'} (\varepsilon)|_{\varepsilon \rightarrow 0},  
\end{equation}
for most cost Hamiltonians $H$ and and states $\ket{\psi}$, 
with strict inequalities whenever $\ket{\psi}$ contains
more than one cost component. 
Hence in this case $\langle H\rangle_{\psi'} (\varepsilon)$ is maximized at some point between the interval boundaries. 
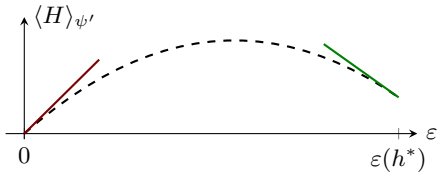
\begin{figure}[htb]
    \begin{center}
        \externalize{cost_expectation_epsilon_sketch}{\begin{tikzpicture}
\begin{axis}[width=7cm,
        height=3.2cm,
        xmin=0.0,
        xmax=1.0,
        xtick={0, 1},
        xticklabels={$0$, $\varepsilon(h^*)$},
        xlabel={$\varepsilon$},
        ymin=0.0,
        ymax=1.5,
        ylabel={$\langle H \rangle_{\psi'}$},
        ytick=\empty,
        hide obscured x ticks=false,
        grid=none,
        enlargelimits,
        axis x line=center,
        axis y line=center,
        enlarge x limits=0.05,
        enlarge y limits=0.05,
        every axis x label/.style={at={(current axis.right of origin)},anchor=west},
        every axis y label/.style={at={(current axis.above origin)},anchor=west}
    ]

    \addplot[name path=G,black,thick,dashed,domain={0:1}]{-4*(x - 0.5)^2 + 1.0 + 0.7*0.7 * x};
    \addplot[name path=G,black!50!red,thick,domain={0.0:0.2}]{5*x};
    \addplot[name path=G,black!50!green,thick,domain={0.8:1.0}]{-3.6*(x - 1) + 0.49};
\end{axis}
\end{tikzpicture}

         }
    \end{center}
    \caption{Illustration of the dependence of the cost expectation value $\langle H \rangle_{\psi'}$ on the choice of the bounds of the spectrum and hence the parameter $\varepsilon$ as per Equation~\eqref{eq:dHde}. For typical cost Hamiltonians the slope is negative at $\varepsilon=\varepsilon(h^*)$ which implies $\langle H \rangle_{\psi'}$ is maximized by some value $\varepsilon\in(0,\varepsilon(h^*))$.
}
    \label{fig:cost_expectation_epsilon_sketch}
\end{figure}
The exact value of optimal $\varepsilon$ is seen to depend on $H$ as well as $\ket{\psi}$. 
For any $H,\ket{\psi}$ outside this case, the optimal value with respect to the specified interval is $\varepsilon(h^*)$. 

In the example of Figure \ref{fig:results_mdqo} the utilization of a 
suboptimal cost bound (and hence smaller choice of $\varepsilon$) indeed improves the expected cost as compared to the tight bound, and this effect is seen to favorably compound over sequences of successful weak measurements. 
The situation for the 
success probability $p_1$ per subsequent step is more subtle.  

Observe that with fixed bound $t>h^*$, for any $\ket{\psi}$ we have $p_1 \leq p_1(\ket{h^*})$. Combining with Eq.~\eqref{eq:probLowerBound} this yields that for any state
\begin{equation}
    \frac12 + \frac12\frac{\langle H\rangle_\psi}{t} \leq \; p_1 \; \leq \frac12 + \frac12 \sin (\frac\pi2 \frac{h^*}{t}).  
\end{equation}
which we observe to be strictly less than $1$ due the the effect of $t$. 
At the same time we see that even with loose cost bounds we always retain the important guarantee that $p_1>1/2$ for superposition states of different cost. Furthmore the conclusion of Remark~\ref{rem:positive-feedback} remains valid. 

\subsection{Measurement-enhanced quantum ansatz} 
\label{sec:postproc}
The preceding analysis shows the reliance 
of the success probability and performance improvement of each weak measurement step
on the initial state. 
Many quantum algorithms for optimization are effectively  
state preparation algorithms, which after a parameter training phase are repeatedly prepared and measured. 
Here we elaborate on how our weak measurement protocol may be combined with  existing quantum circuit ansatze 
to give improved performance over the ansatz alone, with the tradeoffs of modest increase in circuit depth and repetitions; modest in the sense that only a single ancilla qubit is required, and gate costs can be commensurate with an additional layer of an ansatz such as QAOA or its many variants. 

Explicitly, we propose using a given quantum ansatz as input state to Algorithms~\ref{alg:mdqo} and \ref{alg:fcmdqo}. 
Weak measurement steps may be applied to relatively shallow or deep quantum circuit ansatze, that have already undergone a potentially expensive parameter training phase. Embedding in our algorithms facilitates the inclusion of additional circuit elements to improve performance over use of the ansatz alone, while avoiding any parameter retraining.  Hence, in terms of resources, those required for the weak measurement steps again scale as described in Sec.~\ref{sec:resource_est}, but with the additional resources required to prepare the quantum ansatz as initial state. 

In particular, from this perspective even a single successful round of our algorithm as post-processing is beneficial. 
Suppose that using 
a suitable quantum ansatz we have prepared $\ket{\psi}$ such that its expected performance is bounded away from the random guessing value $\langle H\rangle_\psi > \langle H\rangle_0$, which implies $\langle C\rangle_\psi > \langle C\rangle_0$.  
For QAOA in particular it is always possible to find and create such a state~\cite{hadfield2021analytical}. 
Then from \eqref{eq:probLowerBound} we have
 $$p_1 \geq \frac12 + \frac2{\pi} \langle C\rangle_\psi 
 > \frac12 + \frac2{\pi}  \langle C \rangle_0  > 1/2,$$
 such that we see the better the input ansatz state the higher the initial success probability. 
Hence sampling from the improved target distribution introduces an overhead of at most $1/p_1$.
For a single weak-measurement layer 
 this factor is at most $2$.

\paragraph*{Numerical demonstration:} Here we consider application of Algorithm~\ref{alg:mdqo} using the optimized depth-1 QAOA ansatz as input state. Explicitly, for MaxCut we take 
\begin{equation}
    \ket{\psi_0} = \ee^{-\ii \gamma H} \ee^{-\ii \beta \sum_{u\in V} X_u}\ket{+}^{\otimes |V|} \; ,
    \label{eq:qaoa_state}
\end{equation}
with the cost Hamiltonian~$H$ of~\eqref{eq:max_cut_hamiltonian}, and 
angles $\gamma, \beta$ selected as to maximize $\langle H \rangle_{\psi_0}$. 

\begin{figure}[htb]
    \begin{center}
        \externalize{results_mdqo_postproc}{\pgfplotstableread[col sep=comma]{./figures/tikz/data/results_mdqo_postproc.csv}{\expval}

\pgfplotstablegetcolsof{\expval}
\pgfmathsetmacro{\numcols}{\pgfplotsretval}

\pgfmathsetmacro{\lastycol}{\numcols - 1}

\pgfplotscreateplotcyclelist{color list}{
  {gray, fill=gray},
  {black, fill=black},
  {sbblue, fill=sbblue},
  {sbred, fill=sbred},
  {sbgreen, fill=sbgreen},
  {sbpurple, fill=sbpurple},
}
\begin{tikzpicture}[trim left=-1.7em]
    \begin{axis}[
        ybar,
        bar width=2,
        width=8cm,
        height=5.5cm,
        tick style={draw=none}, ymax=0.60,
        ytick={0, 0.2, 0.4, 0.6},
        enlarge x limits=0.10,
        enlarge y limits=0.05,
        xlabel={Cost $h$},
        ylabel={Probability $p(h)$},
        legend style={draw=none,
opacity=0.8,
            inner sep=1pt,
            legend columns=2,
            column sep=0.1ex,
            at={(0.01,0.99)},
            anchor=north west,
        },
grid=both,
        cycle list name=color list,
        axis background/.style={fill=sbbackground},
        axis line style={draw=none},
        grid style={color=white},
    ]

\foreach \i in {1,...,\lastycol} {

\pgfplotstablegetcolumnnamebyindex{\i}\of{\expval}\to\pgfplotsretval
\addplot table [
            x index=0,
            y index=\i,
        ] {\expval};
        \addlegendentryexpanded{\pgfplotsretval}
    }
    \end{axis}
\end{tikzpicture}
         }
    \end{center}
    \caption{\label{fig:results_mdqo_postproc}
        Cost probability density $p(l)$ 
        of Equation~\eqref{eq:costProbDensity} for the MaxCut instance of 
        Figure~\ref{fig:results_mdqo}, with respect to various different states: The uniform random guessing state $\ket{+}^{\otimes |V|}$ (gray), the QAOA$_1$ state by itself~\eqref{eq:qaoa_state} (black), and followed by Algorithm~\ref{alg:mdqo} with several values of $k_1$, with $k_0=0$ here.
        }
\end{figure}
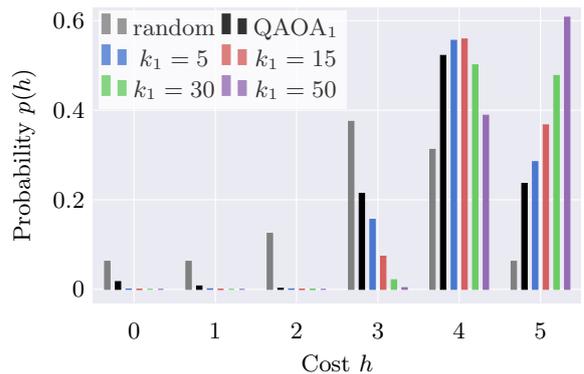
\begin{table}[h]
    \begin{center}
        \begin{filecontents*}{./figures/tikz/data/results_mdqo_postproc_exp_value.csv}
random,QAOA$_1$,$k_1=5$,$k_1=15$,$k_1=30$,$k_1=50$
3.0,3.93,4.12,4.29,4.46,4.6
\end{filecontents*}
\pgfplotstabletypeset[
            every head row/.style={
                before row=\toprule,after row=\midrule},
            every last row/.style={
                after row=\bottomrule},
            col sep=comma,
            ]
            {./figures/tikz/data/results_mdqo_postproc_exp_value.csv} 
    \end{center}
    \caption{Cost expectation value $\langle H \rangle$ 
    for various different states as in Figure~\ref{fig:results_mdqo_postproc}.
    }\label{tab:results_mdqo_postproc}
\end{table}

Figure~\ref{fig:results_mdqo_postproc} shows the resulting probability distributions over solution costs 
\begin{equation} \label{eq:costProbDensity}
    p(l) := \sum_{\mathbf{x}: h(\mathbf{x})=l} \left| \braket{\psi }{\mathbf{x}}\right|^2
\end{equation}
for the same MaxCut instance as in Figure~\ref{fig:results_mdqo}, over several possible outcomes of Algorithm~\ref{alg:mdqo}.
As expected the QAOA state improves on random guessing and each successful weak-measurement step in Algorithm~\ref{alg:mdqo} pushes the probability distribution further towards the optimal value (which is $h=5$ here), such that the expected approximation ratio achieved is improved (see Table~\ref{tab:results_mdqo_postproc}). 

\subsection{Algorithm Analysis: Multi-Step} \label{sec:multistep}
We now consider the more general case of multiple weak-measurement steps of Algorithm 1. 
Based on the analysis of the previous section we consider target criterion for Algorithm~\ref{alg:mdqo} of success at a significant fraction of steps, which corresponds to favorable amplitude modulation of high cost states. 
We give several results showing that the conditions on the success probability derived above lead to moderately scaling number of overall repetitions to achieve multi-step runs with a target fraction of successes achieved. 

Recall the results of Section~\ref{sec:mdqo} for some fixed input state $\ket{\psi}$,  
with the state $\ket{\psi_{k_0,k_1}}$ resulting from $k=k_0+k_1$ weak measurement steps given in Equation~\eqref{eq:ent-meas-state}. 
Observe that the success probability condition $p_1>1/2$ guarantees that for any sequence of $k$ weak measurement steps, we expect at least $k/2$ successful outcomes, i.e., the probability of more successes than failures is at least $1/2$. 
At the same time, from the positive feedback effect described in 
Remark~\ref{rem:positive-feedback} we expect still much better performance than this lower bound suggests. 
We next use our results to bound the expected number of steps to achieve a given difference between successes and failures.

\begin{prop} \label{lem:allSucc}
Consider a sequence of weak measurement steps of Algorithm 1 applied to a state $\ket{\psi}$, with initial success probability $p_1$ given by Eq.~\eqref{eq:succProb}, where we reset the count and state whenever a failure outcome occurs. Then the expected number of overall steps $K$ to obtain a sequence of $k_1=L$ successes is at most
\begin{equation} \label{eq:allSucc}
    E[K]\leq \frac{1-p_1^{L}}{(1-p_1)p_1^{L}}.
\end{equation}    
\end{prop}
\begin{proof}
    As shown a sequence of successful weak measurements can only increase $p_1$. Hence the expected number of trials is at most the value corresponding to the case where $p_1$ remains fixed, which is the right hand side of \eqref{eq:allSucc}, and is a standard formula corresponding to a run of $k_1$ heads of a biased coin.
\end{proof}
For $p_1 \rightarrow 1$ this gives $E[K]=L$ as expected. Clearly $E[K]\geq L$ always. 
If one has an upperbound~$p$ to the single step success probability after $L$ successes (which can be derived explicitly from Eq.~\ref{eq:succProb}), 
then the same expression as on the righthand side of Equation~\eqref{eq:allSucc} with $p$ instead of $p_1$ gives a lower bound for $E[k]$.
For $p_1\leq p < 1$ this implies 
$E[K]$ grows exponentially with $L$.

When $p_1$ is initially fairly close to $1$, or when $L$ is moderate, this overhead is modest and so we can strictly require $L$ successes, i.e. obtain the state corresponding to $(k_0,k_1)=(0,L)$. When we lack such guarantees we can instead seek to create states with high ratios of successes to failures. 

\begin{prop} \label{lem:expecRep}
    Consider a sequence of weak measurement steps, which produces $k_0,k_1$ many $0,1$ outcomes, respectively.
Let the probability of success at each step be at least some fixed $p>1/2$. 
Then the expected number of overall steps $k_L$ to reach a fixed value $k_1-k_0=L$ is at most 
\begin{equation}
  E[k_L]\leq \frac{L}{2p-1}  .
\end{equation}
If $p>1/2+c$ for some constant $c$ then $k_L=O(L)$.
\end{prop}

\begin{proof}
    Pretend for the moment that the probability of success $p$ is fixed at each step, and let $q=1-p$.  
For $k$ steps we have 
$E[k_1] = pk$ and $E[k_0] = qk$ which gives 
$$E[k_1-k_0] = (p-q)k = (2p-1)k,$$
and so the expected number of trials $k_L$ to reach a fixed value $k_1-k_0=L$ is 
$k_L=\frac{L}{2p-1}$.

With the equal probability assumption relaxed as in our algorithm we then have
$$E[k_1-k_0] = \sum_{i=1}^k (p_i-q_i) = \sum_{i=1}^k (2p_i-1) \geq (2p-1)k$$
from which the claim follows. Finally, plugging in $p=1/2+c$ directly gives $k_L=O(L)$.
\end{proof}

Now consider the case of incorporating a reset condition such that 
if ever $k_1 - k_0$ 
reaches a fixed threshold $-R$ 
we reset the count, restore the initial quantum state, and restart the process. As expected including the reset feature is seen to decrease the expected number of trials to achieve a given target $L$. 

\begin{lemma} \label{lem:algWithReset}
    Consider a sequence of weak measurement steps, which produces $k_0,k_1$ many $0,1$ outcomes, respectively. 
    Additionally assume we reset the process and counts if ever $k_1-k_0 = - R$ for some $R\in\mathbb{N}$.
    Let the probability of success at each step be at least $p>1/2$, with $q:=1-p$. 
    Then the expected number of overall steps $k_L$ to reach a fixed value $k_1-k_0=L$ is at most 
    \begin{equation} \label{eq:trials_w_reset}
        E[k_L]\leq\; \frac{L}{2p-1} -  \frac{Rq^R}{(p-q)(p^R-q^R)}\left(1-(p/q)^L\right) 
    \end{equation}
    which satisfies $k_L \leq \frac{L}{2p-1}$ for all $R$.
\end{lemma}
The proof is similar to but more general than that of Lemma~\ref{lem:expecRep} and we give the details in Appendix~\ref{app:analysis}.

Observe that for each value of $R$ have $k_L \rightarrow L$ as $p\rightarrow 1$. The case $R\rightarrow\infty$ (i.e., no reset) reproduces the result of Lemma~\ref{lem:expecRep}, which shows that adding reset is always beneficial, as expected.
\begin{cor}
    For the case $R=1$ we have 
    \begin{equation} \label{eq:trials_w_reset1}
     E[k_L]\leq\; \frac{L}{(p-q)} -  \frac{q}{(p-q)^2}\left(1-(p/q)^L\right) 
    \end{equation}
\end{cor}

The results of this section support and reemphasize that 
by shifting the success probability away from $1/2$, we can 
moderate the resource overhead required to achieve a target ratio of successes and correspondingly improved quantum state. 
 
\section{Application to Constrained Optimization}\label{sec:constrainedOpt}
Here we consider the more general setting of optimization problems that additionally come with a set of hard constraints that must be satisfied. 
Bitstrings and corresponding quantum states satisfying all of the hard constraints are referred to as \textit{feasible}. 
We consider two distinct approaches extending our above results, based on penalty terms and feasible subspace preserving operators, respectively. 
In particular we apply the preceding analysis to argue that the latter approach is particularly advantageous. 
To make these ideas clear we start with a motivating example. 

\paragraph*{Example. Maximum Independent Set (MIS):} 
Given a graph $G=(V,E)$ we seek the largest subset of vertices $S\subset V$, $|V|=n$, such that no two vertices in $S$ are connected by an edge in $E$. 
The vertices in $S$ may be indicated by the bits $x\in\{0,1\}^n$ with simple cost function
\begin{equation}
    H=\sum_{u \in V} x_u = \frac{n}2I-\frac12\sum_{u\in V} Z_u \; ,
    \label{eq:mis_hamiltonian}
\end{equation}
subject to the constraint $P=0$ where 
\begin{equation}  \label{eq:mis_penalty}
P:=\sum_{(u,v)\in E} x_u x_v \; .
\end{equation}

Observe that the spectral norm $\|H\|=n$ on the full Hilbert space, but is reduced to $h^*$ on the subspace of feasible states. Moreover, the constraint term $\|P\|=|E|=O(n^2)$ is significantly larger over the full Hilbert space, while it is identically zero when restricted to the feasible subspace. 
As we shall see, these properties directly lead to significantly different Hamiltonian rescaling parameters between the two 
general 
approaches we consider. 

\subsection{Penalty-based approach} 
In penalty-based approaches, 
one utilizes a penalty operator $P$ that is zero on feasible states and positive on infeasible ones (such as Eq.~\eqref{eq:mis_penalty} for the MIS problem), to give the penalized cost Hamiltonian 
\begin{equation}  \label{eq:mis_hamiltonian_with_penalty}
    H' = H - \lambda P \; ,
\end{equation}
where $\lambda$ is a fixed positive parameter called the (global) \emph{penalty weight}.
Penalty-based constructions are known for a wide variety of NP-hard optimization problems~\cite{lucas2014ising,hadfield2018quantum}.
For example, generally, each hard constraint (or appropriate square thereof) may be added to construct~$P$. The penalty weights are typically chosen large enough such that all infeasible states have penalized cost worse than any feasible state in the (target or full) cost subspace. 

We first propose running Algorithm $1$ directly with penalized cost Hamiltonian $H'$, for fixed penalty weight. 
Then from the perspective of our algorithm the same analysis as before applies, with infeasible states simply those of lowest cost with respect to $H'$,  
and successful outcomes driving states toward higher cost distributions as desired. 

The catch, however, with this approach is that, as is common with penalty-based approaches generically, one has to deal somehow with the possibility of sampling infeasible states. One could always discard such outcomes, post-selecting on feasible solutions, at the expense of an increased number of overall runs. For some problems post-processing may be used to methodically map any obtained infeasible strings to feasible ones. (See \cite[App. D]{brady2023iterative} for such post-processing procedure for MIS.) 

A second and perhaps more consequential catch is that for many such problems the norm of 
the penalty term 
can be much larger than that of the 
unpenalized part $H$, resulting in a 
significant
increase to the spectral range of the cost Hamiltonian $H'$. 
For Algorithm~\ref{alg:mdqo} this results in a significantly reduced~$\varepsilon$ parameter, and hence reduced performance improvement per successful step. 

Indeed, for the MIS problem, if the penalized Hamiltonian $H' = H - \lambda P$ is used, then we require $\varepsilon = O(\frac1{\lambda n^2})$, which is significantly smaller than the $\varepsilon = O(1/n)$ scaling required for $H$ alone. 
Further note that while in some quantum algorithm as-large-as-possible penalty weights appear desirable, in our setting this does not appear to be the case as this can exacerbate  the problem described above.

We next explain how instead employing a feasible-subspace invariant approach avoids the need for penalty terms such that one may retain favorable $\varepsilon$ scaling. 

\subsection{Feasible subspace approach}
Here we consider quantum algorithms with dynamics restricted to stay within the feasible subspace. Any (ideal) measurement in the computational basis is then guaranteed to return a feasible bitstring. 
A prototypical example is the quantum alternating operator ansatz constructions of \cite{hadfield2019quantum}, shown there for a variety of problems with hard constraints.

For such algorithms, 
we assume that all intermediate 
quantum states $\ket{\psi}$ lie solely within the feasible subspace, 
i.e.~$\bra{\psi}P\ket{\psi}=0$, 
including the initial state $\ket{\psi_0}$. 
Due to the necessary 
overlap requirements the initial state  
$\ket{\psi_0}$ should be one with support on target states (recall Remark~\ref{rem:bound-support}  above). 
Hence this approach may be similarly applied a feasibility-preserving quantum circuit, such as instantiations of the quantum alternating operator ansatz, or to simpler initially feasible states, for example easily preparable superpositions of feasible states. See in particular the discussion of possible initial feasible states in \cite{hadfield2019quantum}.

Next observe that the weak measurement operator is diagonal with respect to the computational basis and so it also preserves feasibility, i.e.,
the state after the weak measurement operation $W_b$ remains a superposition of feasible states, whichever the measurement outcome. 
Hence, up to noise, the states in the main register never leave the feasibility subspace under sequences of weak measurements as in Algorithm 1.

Hence, here for constrained problems we propose to simply employ Algorithm~\ref{alg:mdqo} with the unpenalized cost Hamiltonian~$H$, 
but applied to states $\ket{\psi_0}$ guaranteed to lie withing the feasible subspace as input. This approach advantageously avoids the detrimental effects of incorporating penalty terms described above. Moreover, we emphasize it is a much milder requirement than other feasible subspace preserving quantum algorithms. For example, feasibility-preserving variants of QAOA or Quantum Annealing require not only feasible initial states, but relatively complicated dynamical operators in addition. Algorithm~\ref{alg:mdqo} avoids the need for the latter entirely. 

For example, for the weak measurement operator is easily implemented with $n$ many 2-local gates. As explained, by using a feasible state $\ket{\psi_0}$, the rescaling parameter $\epsilon$ can be selected as $O(1/n)$ as opposed to $O(1/n^2)$ which yields immediate improvement over the penalty-based approach. 

\subsection{Numerical Comparison }
In Figure~\ref{fig:results_mis_mdqo} we 
compare penalty-based and feasibilty-preserving approaches 
for Algorithm~\ref{alg:mdqo} applied to MIS on the same example graph instance as in Figure~\ref{fig:results_mdqo}. For the penalty-based approach we choose $\lambda=3$ in Eq.~\eqref{eq:mis_hamiltonian_with_penalty} so that feasible states ($h\geq0$) are separated from the infeasible states ($h\leq 0$), and the spectrum hence spans $-13 \leq h \leq 3$.
The upper bounds for the scaling function \eqref{eq:scaling_bounds}, are chosen to be $t=n=5$ for the loose case, and $t=h^*=3$ for the tight bound for comparison. 
For the lower bound, we get $-s=n-\lambda m = -13$ for the penalty-based approach and $-s=0$ for the invariant feasible subspace approach.

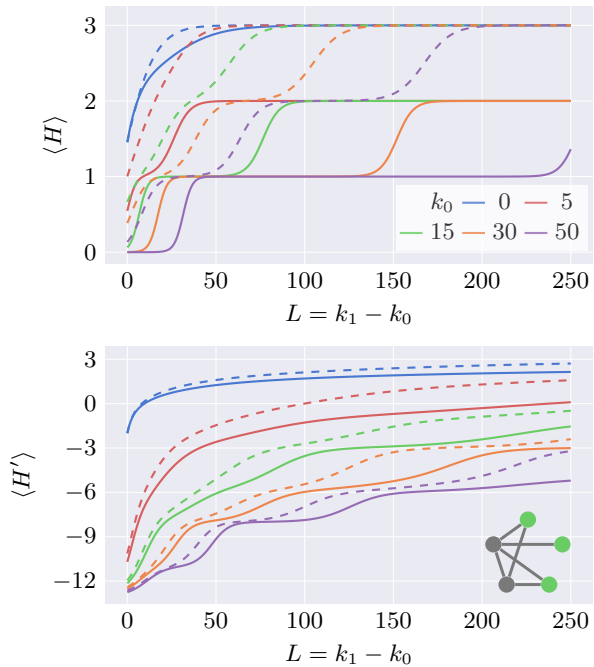
\begin{figure}[htb]
    \begin{center}
        \begin{tikzpicture}
            \node (plot) []{\externalize{results_mis_mdqo}{\pgfplotstableread[col sep=comma]{./figures/tikz/data/results_mis_mdqo_exp_value_best.csv}{\expvalbest}
\pgfplotstableread[col sep=comma]{./figures/tikz/data/results_mis_mdqo_exp_value_worst.csv}{\expvalworst}
\pgfplotstableread[col sep=comma]{./figures/tikz/data/results_mis_mdqo_penalty_exp_value_best.csv}{\expvalpenbest}
\pgfplotstableread[col sep=comma]{./figures/tikz/data/results_mis_mdqo_penalty_exp_value_worst.csv}{\expvalpenworst}

\pgfplotstablegetcolsof{\expvalbest}
\pgfmathsetmacro{\numcols}{\pgfplotsretval}

\pgfmathsetmacro{\lastycol}{\numcols - 1}

\pgfplotscreateplotcyclelist{color list}{
  {sbblue, solid, thick},
  {sbred, solid, thick},
  {sbgreen, solid, thick},
  {sborange, solid, thick},
  {sbpurple, solid, thick},
{sbblue, dashed, thick},
  {sbred, dashed, thick},
  {sbgreen, dashed, thick},
  {sborange, dashed, thick},
  {sbpurple, dashed, thick},
}
\begin{tikzpicture}[trim axis left]
    \begin{groupplot}[
        group style={group size=1 by 2, 
vertical sep=8ex,
            ylabels at=edge left,
        },
    ]
    \nextgroupplot[width=8cm,
        height=5cm,
tick style={draw=none}, ymax=3.1,
        ytick={0, 1, 2, 3},
        enlarge x limits=0.05,
        enlarge y limits=0.05,
        xlabel={$L=k_1 - k_0$},
        ylabel={$\langle H \rangle$},
        legend style={draw=none,
opacity=0.8,
            inner sep=1pt,
            legend columns=3,
            column sep=0.1ex,
            at={(0.99,0.05)},
            anchor=south east,
        },
        legend image post style={xscale=0.5}, grid=both,
        cycle list name=color list,
        axis background/.style={fill=sbbackground},
        axis line style={draw=none},
        grid style={color=white},
    ]

    \addlegendimage{empty legend}
    \addlegendentry[]{$k_0$}
\foreach \i in {1,...,\lastycol} {

\pgfplotstablegetcolumnnamebyindex{\i}\of{\expvalbest}\to\pgfplotsretval
\addplot table [
            x index=0,
            y index=\i,
        ] {\expvalbest};
        \addlegendentryexpanded{\pgfplotsretval}
    }
    \foreach \i in {1,...,\lastycol} {

\pgfplotstablegetcolumnnamebyindex{\i}\of{\expvalworst}\to\pgfplotsretval
\addplot table [
            x index=0,
            y index=\i,
            dashed,
        ] {\expvalworst};
    }
    \nextgroupplot[width=8cm,
        height=5cm,
        ymax=3.1,
        ytick={-12, -9, -6, -3, 0, 3},
tick style={draw=none}, enlarge x limits=0.05,
        enlarge y limits=0.05,
        xlabel={$L=k_1 - k_0$},
        ylabel={$\langle H' \rangle$},
grid=both,
        cycle list name=color list,
        axis background/.style={fill=sbbackground},
        axis line style={draw=none},
        grid style={color=white},
    ]

\foreach \i in {1,...,\lastycol} {

\pgfplotstablegetcolumnnamebyindex{\i}\of{\expvalpenbest}\to\pgfplotsretval
\addplot table [
            x index=0,
            y index=\i,
        ] {\expvalpenbest};
}
    \foreach \i in {1,...,\lastycol} {

\pgfplotstablegetcolumnnamebyindex{\i}\of{\expvalpenworst}\to\pgfplotsretval
\addplot table [
            x index=0,
            y index=\i,
            dashed,
        ] {\expvalpenworst};
    }
    \end{groupplot}
\end{tikzpicture}
                 }
            };
            \node (graph) [left=1.3em of plot.south east,
                yshift=8.0ex,
                anchor=south east,
                scale=0.7,
                ]{\externalize{results_mis_mdqo_graph}{
\tikzset{%
    vertex/.style={%
        circle,
        draw=none,
        minimum width=1ex,
    },
    edge/.style={%
        ultra thick,
        sbgray,
    },
}
\def\radius{5ex}
\begin{tikzpicture}
    \node (1) [vertex, fill=sbgreen] at ({360/5*(1.25)}:\radius) {};
    \node (2) [vertex, fill=sbgray] at ({360/5*(2.25)}:\radius) {};
    \node (3) [vertex, fill=sbgray] at ({360/5*(3.25)}:\radius) {};
    \node (4) [vertex, fill=sbgreen] at ({360/5*(4.25)}:\radius) {};
    \node (5) [vertex, fill=sbgreen] at ({360/5*(5.25)}:\radius) {};

    \draw[edge] (1) -- (2);
    \draw[edge] (2) -- (3);
    \draw[edge] (3) -- (4);
    \draw[edge] (1) -- (3);
    \draw[edge] (2) -- (4);
    \draw[edge] (2) -- (5);

\end{tikzpicture}
                     }
            };
        \end{tikzpicture}
    \end{center}
    \caption{Comparison of invariant feasible subspace approach (top) and penalty-based (bottom) approach:
        Cost expectation value of the MIS Hamiltonian (\eqref{eq:mis_hamiltonian} and~\eqref{eq:mis_hamiltonian_with_penalty}, respectively) for the displayed graph with respect to the state in \eqref{eq:ent-meas-state} against the difference in desired outcomes and undesired outcomes $L=k_1-k_0$ for various fixed numbers of undesired outcomes $k_0$, 
        for the infeasible subspace approach (top) and the penalty-based approach with $\lambda=3$ (bottom).
        The solid lines and dashed lines indicate the tight $t=h^*$ and loose upper bounds $t=n=|V|$, respectively.
        }\label{fig:results_mis_mdqo}
\end{figure}
As anticipated, the invariant feasible subspace approach outperforms the penalty-based approach.
This is mostly due to the broadening of the spectral range in the penalty-based approach which leads to a smaller $\varepsilon$. 
Note that the behavior of the success probability (not shown) is similar to the one in Figure~\ref{fig:results_mdqo}, i.e., the loose bound performs generally better than the tight bound, for both approaches.
In Figure~\ref{fig:results_mis_mdqo_penalty} we present additional results for varying strengths of penalty weights. 
As expected, larger penalty weights are more efficient in suppressing infeasible states and therefore exhibit improved performance compared to smaller penalty weights. 

\begin{figure}[htb]
    \begin{center}
        \begin{tikzpicture}
            \node (plot) []{\externalize{results_mis_mdqo_penalty}{\pgfplotstableread[col sep=comma]{./figures/tikz/data/results_mis_mdqo_exp_value_with_penalty_lambda5.csv}{\expvalwithfirst}
\pgfplotstableread[col sep=comma]{./figures/tikz/data/results_mis_mdqo_exp_value_with_penalty_lambda3.csv}{\expvalwithsecond}
\pgfplotstableread[col sep=comma]{./figures/tikz/data/results_mis_mdqo_exp_value_with_penalty_lambda1.csv}{\expvalwiththird}
\pgfplotstableread[col sep=comma]{./figures/tikz/data/results_mis_mdqo_exp_value_just_penalty_lambda5.csv}{\expvaljustfirst}
\pgfplotstableread[col sep=comma]{./figures/tikz/data/results_mis_mdqo_exp_value_just_penalty_lambda3.csv}{\expvaljustsecond}
\pgfplotstableread[col sep=comma]{./figures/tikz/data/results_mis_mdqo_exp_value_just_penalty_lambda1.csv}{\expvaljustthird}

\pgfplotstablegetcolsof{\expvalwithfirst}
\pgfmathsetmacro{\numcols}{\pgfplotsretval}

\pgfmathsetmacro{\lastycol}{\numcols - 1}

\pgfplotscreateplotcyclelist{color list}{
  {sbblue, solid, thick},
  {sbred, solid, thick},
  {sbgreen, solid, thick},
  {sbblue, densely dashed, thick},
  {sbred, densely dashed, thick},
  {sbgreen, densely dashed, thick},
  {sbblue, densely dotted, thick},
  {sbred, densely dotted, thick},
  {sbgreen, densely dotted, thick},
}
\begin{tikzpicture}[trim axis left]
    \begin{groupplot}[
        group style={group size=1 by 2, 
vertical sep=8ex,
            ylabels at=edge left,
        },
    ]
    \nextgroupplot[width=8cm,
        height=5cm,
tick style={draw=none}, ymax=3.1,
        ytick={-12, -9, -6, -3, 0, 3},
        enlarge x limits=0.05,
        enlarge y limits=0.05,
        xlabel={$L=k_1 - k_0$},
        ylabel={$\langle H - \lambda P \rangle$},
        ylabel shift=-1.5ex,
        legend style={draw=none,
opacity=0.8,
            inner sep=1pt,
            legend columns=2,
            column sep=0.1ex,
            at={(0.99,0.05)},
            anchor=south east,
        },
        legend image post style={xscale=0.5}, grid=both,
        cycle list name=color list,
        axis background/.style={fill=sbbackground},
        axis line style={draw=none},
        grid style={color=white},
    ]

    \addlegendimage{empty legend}
    \addlegendentry[]{$k_0$}
\foreach \i in {1,...,\lastycol} {

\pgfplotstablegetcolumnnamebyindex{\i}\of{\expvalwithfirst}\to\pgfplotsretval
\addplot table [
            x index=0,
            y index=\i,
        ] {\expvalwithfirst};
        \addlegendentryexpanded{\pgfplotsretval}
    }
    \foreach \i in {1,...,\lastycol} {

\pgfplotstablegetcolumnnamebyindex{\i}\of{\expvalwithsecond}\to\pgfplotsretval
\addplot table [
            x index=0,
            y index=\i,
            dashed,
        ] {\expvalwithsecond};
    }
    \foreach \i in {1,...,\lastycol} {

\pgfplotstablegetcolumnnamebyindex{\i}\of{\expvalwiththird}\to\pgfplotsretval
\addplot table [
            x index=0,
            y index=\i,
            dashed,
        ] {\expvalwiththird};
    }
    \nextgroupplot[width=8cm,
        height=5cm,
        ymax=6.1,
        ytick={0, 1, 2, 3, 4, 5, 6},
tick style={draw=none}, enlarge x limits=0.05,
        enlarge y limits=0.05,
        xlabel={$L=k_1 - k_0$},
        ylabel={$\langle P \rangle$},
legend style={draw=none,
opacity=0.8,
            inner sep=1pt,
            legend columns=1,
            column sep=0.1ex,
            at={(0.99,0.99)},
            anchor=north east,
        },
        legend image post style={xscale=0.5}, grid=both,
        cycle list name=color list,
        axis background/.style={fill=sbbackground},
        axis line style={draw=none},
        grid style={color=white},
    ]

    \addlegendimage{solid, thick}
    \addlegendentry[]{$\lambda=5$}
    \addlegendimage{dashed, thick}
    \addlegendentry[]{$\lambda=3$}
    \addlegendimage{densely dotted, thick}
    \addlegendentry[]{$\lambda=1$}
\foreach \i in {1,...,\lastycol} {

\pgfplotstablegetcolumnnamebyindex{\i}\of{\expvaljustfirst}\to\pgfplotsretval
\addplot table [
            x index=0,
            y index=\i,
        ] {\expvaljustfirst};
}
    \foreach \i in {1,...,\lastycol} {

\pgfplotstablegetcolumnnamebyindex{\i}\of{\expvaljustsecond}\to\pgfplotsretval
\addplot table [
            x index=0,
            y index=\i,
            dashed,
        ] {\expvaljustsecond};
    }
    \foreach \i in {1,...,\lastycol} {

\pgfplotstablegetcolumnnamebyindex{\i}\of{\expvaljustthird}\to\pgfplotsretval
\addplot table [
            x index=0,
            y index=\i,
            dashed,
        ] {\expvaljustthird};
    }
    \end{groupplot}
\end{tikzpicture}
                 }
            };
            \node (graph) [right=1.3em of plot.center,
                yshift=8.0ex,
                anchor=south west,
                scale=0.7,
                ]{\externalize{results_mis_mdqo_penalty_graph}{
\tikzset{%
    vertex/.style={%
        circle,
        draw=none,
        minimum width=1ex,
    },
    edge/.style={%
        ultra thick,
        sbgray,
    },
}
\def\radius{5ex}
\begin{tikzpicture}
    \node (1) [vertex, fill=sbgreen] at ({360/5*(1.25)}:\radius) {};
    \node (2) [vertex, fill=sbgray] at ({360/5*(2.25)}:\radius) {};
    \node (3) [vertex, fill=sbgray] at ({360/5*(3.25)}:\radius) {};
    \node (4) [vertex, fill=sbgreen] at ({360/5*(4.25)}:\radius) {};
    \node (5) [vertex, fill=sbgreen] at ({360/5*(5.25)}:\radius) {};

    \draw[edge] (1) -- (2);
    \draw[edge] (2) -- (3);
    \draw[edge] (3) -- (4);
    \draw[edge] (1) -- (3);
    \draw[edge] (2) -- (4);
    \draw[edge] (2) -- (5);

\end{tikzpicture}
                     }
            };
{};
        \end{tikzpicture}
    \end{center}
    \caption{Influence of various penalty weights $\lambda$ onto the performance of the penalty-based:
        Cost expectation value of the full MIS Hamiltonian~\eqref{eq:mis_hamiltonian_with_penalty} (top) and just the penalty term $P$ (bottom) for the displayed graph with respect to the state in \eqref{eq:ent-meas-state} against the difference in desired outcomes and undesired outcomes $L=k_1-k_0$ for various fixed numbers of undesired outcomes $k_0$.
        The solid, dashed and dotted lines are for $\lambda=5, 3, 1$, respectively.
        }\label{fig:results_mis_mdqo_penalty}
\end{figure}
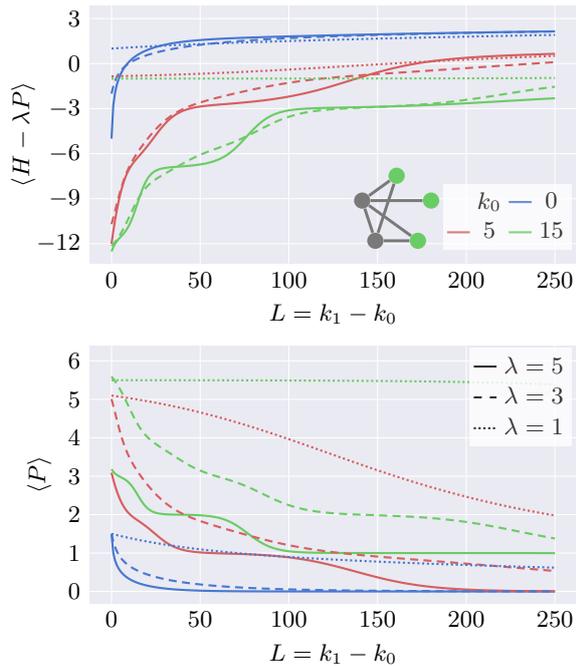
 
\section{Feedback-controlled Measurement-driven Quantum Optimization}\label{sec:fcmdqo}
Here we consider a further generalization of Algorithm~\ref{alg:mdqo} inspired by the approach of \cite{mao2023measurement} as well as existing quantum approaches to optimization such as QAOA~\cite{farhi2014quantum,hadfield2019quantum}. Both the approaches of \cite{mao2023measurement} and Algorithm 2 below additionally assume the ability to perform relatively fast mid-circuit measurement followed by feedback such that the outcomes can be used to dynamically control future gates in the circuit. 
(We emphasize that the control requirements of Algorithm~\ref{alg:mdqo} are significantly weaker.) 
In \cite{mao2023measurement}, as the specific goal considered was obtaining the ground state of a quantum Hamiltonian, additional steps were required to ensure the repeated weak measurement algorithm did not get stuck at a suboptimal eigenstate. In particular, after each weak measurement step a generic \emph{scrambling} operator is conditionally applied based on 
a classical function of 
the sequence of intermediate measurement outcomes obtained so far. We import a similar extension to our weak measurement approach,  motivated by the desire to escape potential configurations where Algorithm~\ref{alg:mdqo} could become stuck or suffer from slow improvement. For this we directly leverage the theory of mixing operators from QAOA, which serve a similar purpose algorithmically, and for which a variety of different operators suitable to different types of problems have been proposed in the literature. In particular we leverage the feasible-subspace-preserving constructions of \cite{hadfield2019quantum} to show how penalty terms and their overhead can be again advantageously avoided for many constrained optimization problems. 
Hence our approach is immediately applicable to both constrained and unconstrained settings. 

Algorithm~\ref{alg:fcmdqo} is referred to as  \ac{algo} and is depicted in Figure~\ref{fig:flow-fcmdqo}.
\begin{algorithm}[H]
    \textbf{Input}
\begin{itemize}[leftmargin=2em]
    \item 
        Problem Hamiltonian $H$ encoding a pseudo-boolean cost function $h(x)$ to maximize, possibly with respect to a set of hard constraints that must be satisfied
    \item
        Rescaled cost Hamiltonian 
        $0\leq C:=\mathcal{E}(H) \leq \frac{\pi}4$, possibly including penalty terms
    \item 
        Initial $n$-qubit state  $\ket{\psi_0}$, 
        may or may not be guaranteed to lie in feasible subspace
    \item 
        Scrambling unitary $U_\mathrm{S}$ (see Sec.~\ref{sec:fcmdqo-scrambling}) 
    \item 
        Return criteria 
        $\mathcal{R}(\mathbf{b},\dots)$
        depending on the measurement outcomes $\mathbf{b}$ 
        and other parameters
    \item 
        Scrambling criteria 
        $\mathcal{S}(\mathbf{b}, \dots)$
        depending on the measurement outcomes $\mathbf{b}$ 
        and other parameters
\end{itemize}
\textbf{Output} 
\begin{itemize}[leftmargin=2em]
    \item 
        Bitstring giving approximate solution to $h(\mathbf{x})$
\end{itemize}
\begin{algorithmic}[1]
    \State $\ket{\psi} \gets \ket{\psi_0}$ 
    \State $\mathbf{b} \gets \emptyset$
    \While{$\neg \; \mathcal{R}(\mathbf{b},\dots)$} 
        \State Introduce fresh ancilla qubit in $\ket{+}$ state
            \label{alg:line:fcmdqo_fresh_ancilla}
        \State 
            Apply entangling gate 
            $\ket{\mathbf{\Psi}} \gets \ee^{-\ii C \otimes Y}\ket{\psi} \otimes \ket{+}$
            \label{alg:line:fcmdqo_entangling}
        \State 
            \begin{varwidth}[t]{0.9\textwidth}
                Measure ancilla qubit in computational basis with outcome $b \in \{0, 1\}$ \label{alg:fcmdqo:line:measurement}
            \end{varwidth}
            \label{alg:line:fcmdqo_measurement}
        \State 
            \begin{varwidth}[t]{0.8\textwidth}
                Increment measurement outcome counters\\ $k_a = k_a + \delta_{ab}$, for $a \in \{0, 1\}$
            \end{varwidth}
        \If{$\mathcal{S}(\mathbf{b}, \dots)$
        } 
            \State 
                Apply scrambling operation 
                $\ket{\psi} \gets U_\mathrm{S} \ket{\psi}$
                \label{alg:line:fcmdqo_scrambling}
            \State 
                Reset 
                outcome count $(k_0, k_1) \gets (0, 0)$
            \label{alg:line:fcmdqo_reset}
        \EndIf
    \EndWhile 
    \State 
        \begin{varwidth}[t]{0.9\textwidth}
            Measure $\ket{\psi}$ in computational basis and return resulting $n$-bit string
        \end{varwidth}
\end{algorithmic}
     \caption{\acf{algo}}\label{alg:fcmdqo}
\end{algorithm}
\begin{figure}
    \begin{center}
        \externalize{flow_chart_fcmdqo}{\tikzset{base/.style={draw=none,
        align=center,
        text centered,
    },
    box/.style={base,
        rectangle,  
        rounded corners, 
        fill=sbblue!50!white,
    },
    condition/.style={base,
        draw=none,
        diamond,
        fill=sbgreen!50!white,
        rounded corners,
    },
    measurement/.style={base,
        trapezium, 
        trapezium left angle = 65,
        trapezium right angle = 115,
        trapezium stretches,
        fill=sbred!50!white,
        rounded corners,
    },
    arr/.style={->,
        line width=1pt,
        color=black!90!white,
        -latex,
    },
    line/.style={line width=1pt,
        color=lightgray,
    },
    label/.style={pos=0.5,
        fill=white,
        inner sep=0.5ex,
    },
}
\begin{tikzpicture}[font=\small,thick]

\node (init) [box,
    ] {Init\\
        $\ket{\psi} \leftarrow \ket{\psi_0}$
    };
\node (while) [condition,
    below=5ex of init,
    ] {while
        $
        \neg \mathcal{R}
        $
    };
\node (freshancilla) [box,
    below=7ex of while,
    ] {Fresh ancilla\\
        $
        \ket{\psi}
        \otimes
        \ket{+} 
        $
    };
\node (entangling) [box,
    below=5ex of freshancilla,
    ] {Entangling\\
        $
        \ee^{-\ii C \otimes Y}
        \ket{\psi}
        \otimes
        \ket{+} 
        $
    };
\node (midmeasure) [measurement,
    below=5ex of entangling,
    ] {Measure ancilla
    };
\node (conditionscrambling) [condition,
    below=5ex of midmeasure,
    ] {If
        $
        \mathcal{S}
        $
    };
\node (scrambling) [box,
    below=5ex of conditionscrambling,
    xshift=5em,
    ] {Scrambling\\
        $
        U_\mathrm{S}
        $
    };
\node (noscrambling) [inner sep=0pt,
    draw=none,
    at=(conditionscrambling |- scrambling),
    xshift=-4em,
    ] {};
\node (measure) [measurement,
    text width=6em,
    minimum width=5em,
    right=4em of while,
    ] {Measure main register
    };
\node (result) [box,
    fill=white,
    text width=6em,
    minimum width=5em,
    below=5ex of measure,
    ] {bitstring $\mathbf{x}$
    };

\draw[arr] (init) -- (while);
\draw[arr] (while) -- (freshancilla) node[label]{Yes};
\draw[arr] (while) -- (measure) node[label]{No};
\draw[arr] (freshancilla) -- (entangling);
\draw[arr] (entangling) -- (midmeasure);
\draw[arr] (measure) -- (result);
\draw[arr] (midmeasure) -- (conditionscrambling) node[label, pos=0.4]{$b$};
\draw[arr] (conditionscrambling) -- (scrambling) node[label]{Yes};
\draw[arr] (conditionscrambling) -- (noscrambling) node[label]{No};
\draw[arr] (scrambling) -- (noscrambling);
\draw[arr, rounded corners] (noscrambling.south)
    -- ++(-2em, 0)
    |- (while.west)
    node[label, pos=0.25]{$b$}
    ;
\end{tikzpicture}

         }
    \end{center}
    \caption{Flow chart for Algorithm~\ref{alg:fcmdqo}}\label{fig:flow-fcmdqo}
\end{figure}
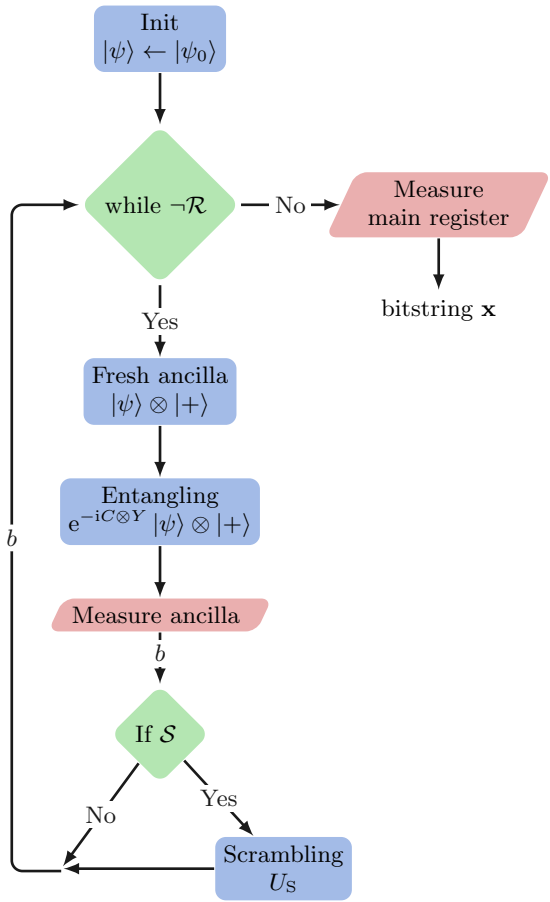
Algorithm~\ref{alg:fcmdqo} generalizes  Algorithm~\ref{alg:mdqo}. 
Similarly to Algorithm~\ref{alg:mdqo} a weak measurement operation is repeatedly applied to a main register (lines~\ref{alg:line:fcmdqo_fresh_ancilla} through~\ref{alg:line:fcmdqo_measurement}). As explained for constrained problems the initial state may be assumed to lie in the feasible subspace, or this assumption relaxed when penalty terms are included in the cost Hamiltonian.
The main extension is the 
scrambling
and count reset operations (lines~\ref{alg:line:fcmdqo_scrambling} and~\ref{alg:line:fcmdqo_reset}), which are applied whenever the scrambling condition $\mathcal{S}$ is meet.
The quantum circuit for the core part of the algorithm is depicted in Figure~\ref{fig:circuit_mdqo2}.
The return criteria $\mathcal{R}$ can be directly imported from Algorithm~\ref{alg:mdqo}, as detailed in Section~\ref{sec:mdqo_algo_details}.
Similarly, adaptive algorithm variants discussed there are directly applicable to Algorithm~\ref{alg:mdqo} also. 
We next elaborate 
on the scrambling condition and operator.
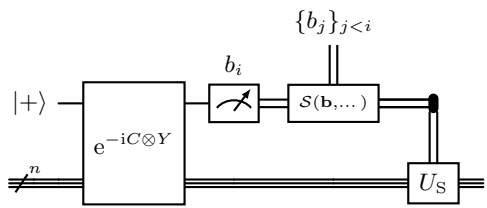
\begin{figure}[htb]
    \begin{center}
        \externalize{circuit_fcmdqo}{\begin{quantikz}[wire types={n,q,b}, column sep=1em, classical gap=0.05cm]
    &
    &
    &
    &
    & \gate[style={draw=none}]{\{b_j\}_{j<i}}
    &
    &
    &
    \\
    \setwiretype{n}
    &
    & \lstick{$\ket{+}$}
    & \gate[2]{\ee^{-\ii C \otimes Y}} \setwiretype{q}
    & \meter{b_i} 
    & \gate{\scriptstyle \mathcal{S}(\mathbf{b},\dots)} \setwiretype{c} \wire[r][1]{c} \wire[u][1]{c}
    & \ctrl[vertical wire=c]{1} 
    & \setwiretype{n}
    &
    \\
    & \qwbundle{n}
    &
    &
    &
    &
    & \gate{U_\mathrm{S}}
    &
\end{quantikz}
         }
    \end{center}
    \caption{Quantum circuit for $i$-th iteration in Algorithm~\ref{alg:fcmdqo}. The ancilla qubit (top) is prepared in a $\ket{+}$ state before it is entangled with the main register (below) and measured.
        The measurement outcome $b_i$ together with the outcomes of previous iterations $\{b_j\}_{j<i}$ informs the scrambling condition $\mathcal{S}$, which controls the scrambling operation $U_\mathrm{S}$. 
    }\label{fig:circuit_mdqo2}
\end{figure}

\subsection{Scrambling criteria} Suppose we are initially given, or that Algorithm~\ref{alg:mdqo} has reached, a state $\ket{\psi}$ that is very close to an eigenstate of $H$. Then, recalling Proposition 2 and Remark 4, subsequent successful weak measurement steps will change the state and hence improve the cost expectation very little. Moreover many successes will be needed to \enquote{escape} such a scenario, which becomes impossible if an exact eigenstate is reached. A related but distinct issue is the requirement of support on target (i.e., high cost) states to begin with, recall Remark 1. 
To alleviate these issues, a well-motivated approach is to incorporate additional operators applied to the quantum state which do not commute with the cost Hamiltonian, and so have the effect of shifting amplitude between different cost eigenstates. 
We propose to leverage the observation that this criterion is closely related to that of mixing operators in QAOA which are designed to perform a similar task  (cf.~design criteria of \cite{hadfield2019quantum}), and so can be adapted to scrambling operators. 

Mixing operators alone do not generally improve the cost expectation, and so can be applied sparingly as compared to the weak measurement steps. 
Hence we propose to apply these operator conditionally, based on a classically computed function of the measurement outcomes obtained so far, denoted by 
$\mathcal{S}(\mathbf{b}, \dots)$. 
Specifically, this function estimates the peak cost of the amplitude modulation function (cf.~Eq.~\eqref{eq:peak-position}) as well as tries to detect if the state has become stuck due to being too close to a particular eigenstate. 
When triggered the mixing operator is applied, the success and failure counts reset (i.e., the current state effectively becomes the new initial state), and the weak measurement steps resumed. 

Different instantiations of the scrambling condition function are possible. 
The implementation we primarily consider is scrambling
is triggered whenever the peak position falls below a fixed value after some minimum number~$\ell$ of weak measurement steps,
\begin{equation}
\mathcal{S}(\mathbf{b}, T, \ell)
    =\;
    (P(k_0, k_1)
    <  \mathcal{E}(T)) \wedge (k_0 + k_1 \geq \ell)  \;
    \label{eq:scrambling_criterion_default}
\end{equation}
where the threshold parameter $T$ satisfies 
$0 \leq \mathcal{E}(T) \leq \frac{\pi}{4}$,
and corresponds to a low-quality solution cutoff (cf.~\cite{mao2023measurement}). 
In practice one could always for example, select $T$ based on the uniform random guessing value, 
or the value of the best solution obtained so far, as to target Algorithm~\ref{alg:fcmdqo} to do at least that well or better
(cf.~adaptive algorithm variants as discussed in Section~\ref{sec:mdqo_algo_details}). 
When the input state $\ket{\psi}$ is 
a parameterized quantum circuit,   
as explained in Section~\ref{sec:postproc},
we often will have existing knowledge of $\langle H \rangle_\psi$ which can similarly be used to determine $T$. 

\subsection{Scrambling operators}\label{sec:fcmdqo-scrambling}
For constrained optimization, we require a scrambling operator $U_\mathrm{S}$ that mixes between (all) feasible states,  
\begin{equation}
    \bra{x} U_\mathrm{S} \ket{x'} \neq 0 \quad \forall x, x' \text{ feasible bitstrings},
    \label{eq:condition_scambling_operator}
\end{equation}
which reduces to all possible bitstrings in the unconstrained case.
This criterion guarantees that if the quantum state lies in the feasible subspace before the application of the scrambling operator, then it is guaranteed to retain the property afterwards. As the weak measurement operator also preserves feasibilty, together this ensures that Algorithm~\ref{alg:fcmdqo} only returns feasible solution samples.  
We emphasize that operators satisfying this criterion are not unique and come with different tradeoffs~\cite{hadfield2019quantum}. 
While in general we are not restricted to scrambling operators imported from QAOA, these operators provide an immediate library for conveniently applying Algorithm~\ref{alg:fcmdqo} to a variety of problems. 

Concretely, for example, for the unconstrained MaxCut problem on a graph $G=(V,E)$ we could choose 
$U_\mathrm{S} = \prod_{u \in V} \exp\left(- \ii \chi X_u \right)$ for some parameter $\chi \in [0, \pi]$. 
For the constrained MIS problem one could use the same mixer if penalty terms are included in the cost Hamiltonian, or this can be avoided if we restrict to the feasible subspace and chose a mixer such as 
$U_\mathrm{S} = \prod_{u \in V} \Lambda_{\prod_{v \in \text{nbhd(u)}} \neg x_v} \exp\left(- \ii \chi X_v \right)$. Here the latter operator denotes a product of a coherently controlled unitary gates for each vertex, controlled by the states of its neighbors in the problem graph; see~\cite{hadfield2019quantum} for details. 
In the simplest implementation of our protocol a fixed $\chi$ can be used that satisfies the above criterion. For more sophisticated implementations $\chi$ could be optimized as a parameter, or itself determined as a function of the obtained measurement outcomes. 

\paragraph*{Mixing operators for initial state generation:}
For many constrained problems the construction of superpositions of feasible states, as needed for initial states in our approach, is not trivial.
However, similar to the approach proposed in~\cite{hadfield2019quantum}, given a single feasible basis state, a mixing operator as above can be applied to generate such a superposition.
For example in the MIS case, we can always use the zero bitstring $\ket{00..0}$ and apply the mixing operator described in the example above. In this initial application a larger $\chi$ parameter may be suitable than what is used for the scrambling operator steps.

\paragraph*{Resource estimation:}
We emphasize that the classical control of the scrambling operators in Algorithm~\ref{alg:fcmdqo} adds no additional quantum resource overhead. As the resources required for the mixing operators we import are often commensurate or less those required for QAOA phase operators~\cite{hadfield2019quantum}, the discussion of Sec.~\ref{sec:resource_est} leads to a similar conclusion compared to the weak measurement operation. As explained, for a given problem different scrambling operators are possible, with different overheads, as likely different performance tradeoffs. 

\subsection{Numerical Proof-of-Concept}
In order to demonstrate the effectiveness including of the scrambling operation,
here we consider a contrived example, where many weak-measurement steps in Algorithm~\ref{alg:mdqo} with undesired outcomes have resulted in a state stuck far away from the target subspace.
In particular we consider the state corresponding to outcomes $k_0=160, k_1=50$ 
for the MaxCut problem considered above in Section~\ref{sec:numerics_mdqo} (see markers in Figure~\ref{fig:results_mdqo}).
Next we simulate a single application of a scrambling operation $U_\mathcal{S}(\chi)=\prod_{u \in V} \ee^{-\ii \chi X_u}$,
followed by another round of weak-measurement steps with outcomes $\tilde{k}_0, \tilde{k}_1$.
Hence the state is equal (up to normalization) to
\begin{equation}
    W_0^{\tilde{k}_0}
    W_1^{\tilde{k}_1}
    U_\mathcal{S}(\chi)
    W_0^{k_0}
    W_1^{k_1}
    \ket{\psi_0} \,.
    \label{eq:state_after_scambling}
\end{equation}
\begin{figure}[htb]
    \begin{center}
        \externalize{results_fcmdqo}{\pgfplotstableread[col sep=comma]{./figures/tikz/data/results_fdmdqo_vary_chi_exp_value_scrambling.csv}{\varychiscrambling}
\pgfplotstableread[col sep=comma]{./figures/tikz/data/results_fdmdqo_vary_chi_exp_value_noscrambling.csv}{\varychinoscrambling}
\pgfplotstableread[col sep=comma]{./figures/tikz/data/results_fdmdqo_exp_value_scrambling.csv}{\scrambling}
\pgfplotstableread[col sep=comma]{./figures/tikz/data/results_fdmdqo_exp_value_noscrambling.csv}{\noscrambling}

\pgfplotstablegetcolsof{\varychiscrambling}
\pgfmathsetmacro{\numcolstop}{\pgfplotsretval}
\pgfplotstablegetcolsof{\scrambling}
\pgfmathsetmacro{\numcolsbottom}{\pgfplotsretval}

\pgfmathsetmacro{\lastycoltop}{\numcolstop - 1}
\pgfmathsetmacro{\lastycolbottom}{\numcolsbottom - 1}

\pgfplotscreateplotcyclelist{color list bottom}{
  {sbblue, solid, thick},
  {sbred, solid, thick},
  {sbgreen, solid, thick},
  {sborange, solid, thick},
{sbblue, dashed, thick},
  {sbred, dashed, thick},
  {sbgreen, dashed, thick},
  {sborange, dashed, thick},
}
\pgfplotscreateplotcyclelist{color list top}{
  {sbbrown, solid, thick},
  {sbgray, solid, thick},
  {sbblue, solid, thick},
  {sbmagenta, solid, thick},
  {sbpurple, solid, thick},
  {black, dashed, thick},
}
\begin{tikzpicture}[trim axis left]
    \begin{groupplot}[
        group style={group size=1 by 2, 
vertical sep=8ex,
            ylabels at=edge left,
        },
    ]
    \nextgroupplot[width=8cm,
        height=5cm,
tick style={draw=none}, ymax=5.1,
        ytick={0, 1, 2, 3, 4, 5},
        enlarge x limits=0.05,
        enlarge y limits=0.05,
        xlabel={$\tilde{k}_1$},
        ylabel={$\langle H \rangle$},
        legend style={draw=none,
opacity=0.8,
            inner sep=1pt,
            legend columns=3,
            column sep=0.1ex,
            at={(0.59,0.22)},
            anchor=south east,
        },
        legend image post style={xscale=0.5}, grid=both,
        cycle list name=color list top,
        axis background/.style={fill=sbbackground},
        axis line style={draw=none},
        grid style={color=white},
    ]

    \addlegendimage{empty legend}
\addlegendentry[]{$\tilde{\chi}$}
    Loop through all Y-columns (index 1 to lastycoltop)
    \foreach \i in {1,...,\lastycoltop} {

\pgfplotstablegetcolumnnamebyindex{\i}\of{\varychiscrambling}\to\pgfplotsretval
\addplot table [
            x index=0,
            y index=\i,
        ] {\varychiscrambling};
        \addlegendentryexpanded{\pgfplotsretval}
    }
    \foreach \i in {1,...,1} {

\pgfplotstablegetcolumnnamebyindex{\i}\of{\varychinoscrambling}\to\pgfplotsretval
\addplot table [
            x index=0,
            y index=\i,
            dashed,
        ] {\varychinoscrambling};
    }
    \node[font=\footnotesize,
] at (axis cs:60,2.2) {without scrambling};
    \nextgroupplot[width=8cm,
        height=5cm,
tick style={draw=none}, ymax=5.1,
        ytick={0, 1, 2, 3, 4, 5},
        enlarge x limits=0.05,
        enlarge y limits=0.05,
        xlabel={$\tilde{L}=\tilde{k}_1 - \tilde{k}_0$},
        ylabel={$\langle H \rangle$},
        legend style={draw=none,
opacity=0.8,
            inner sep=1pt,
            legend columns=3,
            column sep=0.1ex,
            at={(0.99,0.05)},
            anchor=south east,
        },
        legend image post style={xscale=0.5}, grid=both,
        cycle list name=color list bottom,
        axis background/.style={fill=sbbackground},
        axis line style={draw=none},
        grid style={color=white},
    ]

    \addlegendimage{empty legend}
    \addlegendentry[]{$\tilde{k}_0$}
\foreach \i in {1,...,\lastycolbottom} {

\pgfplotstablegetcolumnnamebyindex{\i}\of{\scrambling}\to\pgfplotsretval
\addplot table [
            x index=0,
            y index=\i,
        ] {\scrambling};
        \addlegendentryexpanded{\pgfplotsretval}
    }
    \foreach \i in {1,...,\lastycolbottom} {

\pgfplotstablegetcolumnnamebyindex{\i}\of{\noscrambling}\to\pgfplotsretval
\addplot table [
            x index=0,
            y index=\i,
            dashed,
        ] {\noscrambling};
    }
    \end{groupplot}
\end{tikzpicture}

         }
    \end{center}
    \caption{Illustration of the efficacy of the scrambling operation:
        Expectation value of $\langle H \rangle$ of the MaxCut Hamiltonian~\eqref{eq:max_cut_hamiltonian} for graph displayed in Figure~\ref{fig:results_mdqo} with respect to the state in \eqref{eq:state_after_scambling} (with $k_0=50, k_1=160$) against the difference in desired outcomes and undesired outcomes $\tilde{L}=\tilde{k}_1-\tilde{k}_0$ 
        for various scrambling parameters $\chi=\frac{\tilde{\chi}}{7}\frac{\pi}{4}$ with fixed $\tilde{k_0}=0$ (top) and various fixed numbers of undesired outcomes $\tilde{k}_0$ and fixed $\tilde{\chi}=3$ (bottom). 
        In both plots the dashed lines indicate the absence of scrambling (i.e.~$\chi=0$).
        The tight bound $B=h^*$ was used to determine the cost Hamiltonian rescaling parameters.
        }
    \label{fig:results_fcmdqo}
\end{figure}
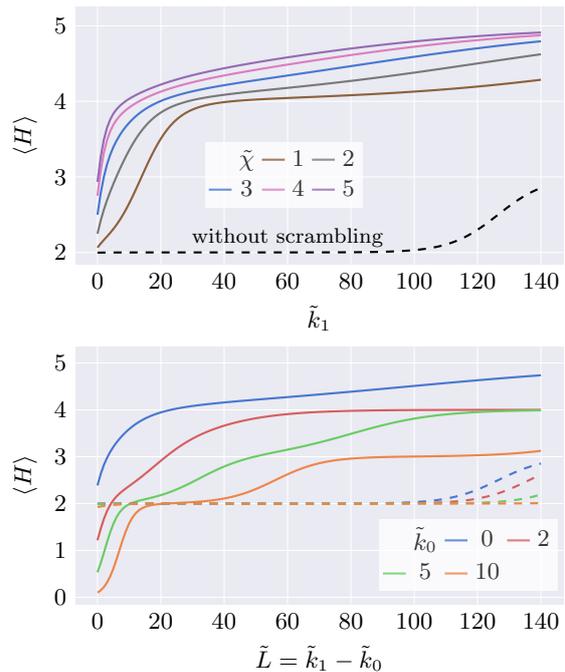
They 
demonstrate that the application of the scrambling operation indeed helps with 
\enquote{escaping} from an otherwise stuck situation.
Additionally, the precise choice of the scrambling parameter $\chi$ is seen to have relative little qualitative effect.
 
\section{Outlook} 
\label{sec:conclusion}
\subsection{NISQ era considerations} As explained, low depth applications of Algorithm~\ref{alg:mdqo} are particularly attractive for near-term quantum devices, i.e., so-called Noisy Intermediate Scale Quantum  (NISQ) devices~\cite{preskill2018quantum}. 
Indeed, the simplest application of Algorithm~\ref{alg:mdqo} requires only a single additional circuit layer and measurement of a single ancilla, with fast feedback not necessary. 

Furthermore, in the approximate optimization setting especially, there is no strict requirement that the weak measurement step is implemented exactly, or even to high precision. Indeed we just require effective operations that (when 'successful')  drive the system register toward an improved energy distribution. 
As we will repeatedly prepare and measure, we ultimately care about the support of the resulting probability distributions, rather than close fidelity to some particular quantum state.
Hence it seems reasonable to conjecture that weak-measurement based quantum optimization may be in some ways more resilient to hardware noise and implementation error than other approaches. 
At the same time, for some types of system noise it should be further advantageous to integrate Algorithms~\ref{alg:mdqo} and~\ref{alg:fcmdqo} as subroutines within a still broader scheme such as Noise-Directed Adaptive Remapping~\cite{maciejewski2024improving} which favorably adapts the algorithm output distribution to the device bias in an iterative manner. 
We save in-depth exploration of the effects of noise and how to deal with it to future work. 

\subsection{Coherent Weak-Measurement Optimization}
Finally, we mention the possibility of developing fully coherent versions of our algorithms, designed to be run at larger depths on fault-tolerant quantum devices of the future. 
Here, instead of a sequence of $r$ weak measurements, $r$ logical ancilla qubits are entangled in turn with the system register and
used to control applications of the scrambling operator. 
The ancilla register is input to quantum subroutine that 
coherently calculates 
scrambling criterion $\mathcal{S}$
such as Equation~\eqref{eq:scrambling_criterion_default}. A sketch of a single step of the such an algorithm is depicted in Figure~\ref{fig:circuit_coherent_mdqo}. 

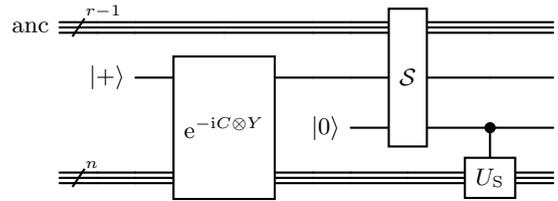
\begin{figure}[htb]
    \begin{center}
        \externalize{circuit_coherent_mdqo}{\begin{quantikz}[wire types={b,q,q,b}, row sep=1ex, classical gap=0.07cm]
    \lstick{\text{anc}}
    & \qwbundle{r - 1}
    &
    &
    & 
    &
    & \gate[3]{\mathcal{S}}
    & 
    & 
    \\
    \setwiretype{n}
    &
    & \lstick{$\ket{+}$}
    & \gate[3]{\ee^{-\ii C \otimes Y}} \setwiretype{q} 
    &
    & 
    & 
    & 
    & 
    \\
    \setwiretype{n}
    &
    &
    &
    &
    & \lstick{$\ket{0}$}
    &
    \setwiretype{q}
    & \ctrl{1} 
    &
    \\
    & \qwbundle{n}
    &
    &
    & 
    &
    & 
    & \gate{U_\mathrm{S}}
    &
\end{quantikz}
         }
    \end{center}
    \caption{Quantum circuit for a single step in the coherent version of Algorithm~\ref{alg:fcmdqo}.  
        For each iteration an ancilla qubit (middle) is needed as before which coherently captures the possible weak measurement outcomes. The ancilla qubits of previous steps (top register) are used for 
            calculating the scrambling criterion $\mathcal{S}$ that 
coherently controls application of the scrambling operation $U_\mathrm{S}$. For the implementation of $\mathcal{S}$ additional scratchpad qubits may be used.
    }\label{fig:circuit_coherent_mdqo}
\end{figure}
Hence one could then in principle embed our algorithms with frameworks such as quantum amplitude amplification to achieve a target configuration $(k_0,k_1)$ quadratically faster than the number of repeated samples requires for Algorithms~\ref{alg:mdqo} and~\ref{alg:fcmdqo}, asymptotically, using an ancilla register of $r=k_0+k_1$ qubits. 

Furthermore this perspective may facilitate 
closer ties and complementary analysis 
of our approach with some of the closely related coherent algorithms discussed in Section~\ref{sec:relatedApproaches}. 
We save the precise details 
for future work.

\subsection{Conclusion} In this work we have developed a 
measurement-driven approach to quantum optimization
based on iterated weak measurements and classical feedback.  
We have shown how to exploit properties unique to combinatorial optimization to rescale and shift the problem as to guarantee good probability of success. 
Furthermore we have explained how our paradigm is applicable to arbitrary initial states, including those produced by other parameterized quantum circuits, as well 
how it naturall accommodates constrained optimization problems. Hence our approach yields both standalone quantum algorithms, as well as \textit{quantum postprocessing}  applicable to a variety of existing quantum circuits. Furthermore, our measurement-driven algorithm may be embedded as subroutines within hybrid solvers such as Iterative Quantum Optimization Algorithms~\cite{brady2023iterative,brady2025quantum}.

As quantum hardware continues to advance, we look forward to experimenting with and further refining our protocols on real-world problems and quantum devices. Indeed, iterative development and refinement guided by extensive empirical testing has been instrumental in the advancement of numerous successful classical heuristics and algorithms. 
As mentioned, the general drive toward fault-tolerant quantum devices has been especially motivating towards the
advancement of 
technologies enabling 
mid-circuit measurement as well as fast feedback control, 
in particular quantum error correction, mitigation, and benchmarking, see e.g.~\cite{chiaverini2004realization,botelho2022error,singh2023mid,koh2024readout,bar2025layered,govia2023randomized,hothem2025measuring}.
A natural question is whether further hardware advancements and features, across all layers of the computational stack, can be similarly leveraged to design the most effective hybrid protocols for combinatorial optimization and beyond. 
 
\section*{Acknowledgements}
We thank our present and past colleagues from NASA QuAIL and Jülich Research Centre and collaborators for numerous helpful related discussions, in particular Nischay Suri, Reuben Tate, 
and Lena Wagner.  
This material is based upon work supported by the U.S. Department of Energy, Office of Science, Office of Advanced Scientific Computing Research under Award Number 89243024SSC000129. We are grateful for support from the NASA Ames Research
Center. 
SH was supported under the Prime Contract No.~80ARC020D0010. TS was funded by the German Federal Ministry of Research, Technology and Space (BMFTR) in the project quantum artificial intelligence for the automotive value chain (QAIAC), Funding No.~13N17166.
 \bibliographystyle{myREVTeX4-2}
\bibliography{references}

\appendix
\section{Supporting Analysis} \label{app:analysis}
\subsection{Useful identities} \label{app:trig}
We make 
use of the following trigonometric identities, in particular in deriving the results of Sec. 2:
\begin{eqnarray}
\sin^2(x)&=&\frac12 -\frac12 \cos(2x)\\
    \sin(a\pm b) &=& \sin a \cos b \pm \cos a \sin b\\
    \cos(a\pm b) &=& \cos a \cos b \mp \sin a \sin b\\
    \cos(x) \pm \sin(x) &=&\sqrt{2} \sin\left(\frac{\pi}{4} \pm x\right)
    \label{eq:trig_ident_cos_minus_sin}\\
    1 \pm \sin(2x) &=& 2 \sin^2\left(\frac{\pi}{4} \pm x\right)
    \label{eq:trig_ident_sin_squared}\\
    \sin(\frac\pi2+x) &=& \cos x\\
\end{eqnarray}
Finally, on $[0,\pi/2]$ we 
have 
$\frac2\pi x \leq \sin(x) \leq x.$
 \subsection{Multi-step algorithm with reset}
Here we give the proof details of Lemma \ref{lem:algWithReset}.

\begin{proof}[Proof of Lemma \ref{lem:algWithReset}]
    First consider the case of fixed probability of success $p$, and define $q:=1-p$.
    Let $E_j$ be the expected number of additional flips to obtain $L$ more $1$s than $0$s from the point $j=k_1-k_0$.
    These quantities are easily seen to satisfy the recurrence relation $E_j = 1 + pE_{j+1} + qE_{j-1}$, 
    for which $E_0$ gives the expected number of trials overall, which we obtain by applying appropriate boundary conditions.

    The target of achieving $L$ more $1$s than $0$s corresponds to the absorbing boundary condition $E_L = 0$. 
    Here if $k_1 - k_0 = R$ the algorithm resets the count and quantum state. 
    This corresponds to the boundary condition $E_{-R} = E_0$. 
    Thus, 
    we have the system of finite difference equations
    $$E_L = 0  $$
    $$E_j= 1+ pE_{j+1} + qE_{j-1}   \;\;\;\; \text{for } j = -R+1, \dots, L-1 $$
    $$E_{-R} = E_0  $$
    which we seek to solve for $E_0$.
    We rearrange the recurrence relation to 
    $E_j - pE_{j+1} - qE_{j-1} = 1 $ 
    for which it is straightforward to derive the homogeneous solutions 
    $E_j = A + B(q/p)^j$, 
    and particular solution $E_j = - \frac{j}{p-q}$ such that the general solution is given by 
    $$E_j = A + B\left(\frac{q}{p}\right)^j - \frac{j}{p-q}, $$
    with coefficients then satisfying $E_0=A+B$. Solving for $A,B$ using the boundary condition yields $E_0$ to be given by the 
    righthand side of Eq.~\eqref{eq:trials_w_reset}. 
    As in Lemma~\ref{lem:expecRep}, as the success probability at each step is $p$ or better we have the bound as claimed. 
    Finally, it is easy to check $k_{L,R=1} \leq k_L \leq k_{L,R-\rightarrow\infty} = \frac{L}{2p-1}$.
\end{proof}
 \subsection{Improvement from reducing $\varepsilon$}
Here we derive Equation~\eqref{eq:dHde} and surrounding claims;
that, counterintuitively, reducing $\varepsilon$ from the one derived from the optimal cost value can sometimes improve the expected cost. In practice we do not know the true optimal value and must use a bound $B> h^*$, which results in a smaller $\varepsilon$ parameter, but nevertheless this can result, surprisingly, in improved performance.
The tradeoff with the improved expected cost is that the single-step success probability decreases with smaller $\varepsilon$. 

For simplicity of arguments assume here that $H\geq 0$ such that we have selected $\alpha=0$ and $C=\varepsilon H$, with $\varepsilon < \varepsilon^*:=\frac{\pi}{4h^*}$, and consider 
the weak measurement of Algorithm 1 applied to a given normalized state $\ket{\psi}$, with a successful outcome yielding the state $\ket{\phi}$. 

From the analysis of Sec.~\ref{sec:mdqo} we may write the single-step success probability as 
$ p_1 (\varepsilon) = \langle \sin^2(\pi/4 + \varepsilon H ) \rangle_\psi,$
and (normalized) cost expectation $\langle H \rangle_{\phi} (\varepsilon) = Q/p_1$ with $Q:= \langle H \sin^2(\pi/4 + \varepsilon H )\rangle_\psi$. 
Then differentiating 
with respect to $\varepsilon$ we have
$$ p'_1:=\frac{d}{d\varepsilon} p_1 = \langle H\sin (\pi/2 + 2\varepsilon H)\rangle_\psi = \langle H \cos 2\varepsilon H\rangle_\psi$$
which implies $p_1'\geq 0$ under the assumptions on~$H,\varepsilon$. 
This shows that slightly decreasing $\varepsilon$ will always have the effect of slightly decreasing the single-step success probability. Nevertheless, as discussed, and demonstrated Figure~\ref{fig:results_mdqo}, 
the compounded effect of using suboptimal bounds can result in increased or decreased single step success probability relative to the $\varepsilon$ corresponding to the optimal cost. 

Similarly, applying the quotient rule gives 
$$ 
\langle H \rangle'_{\phi} = \frac{ Q'}{p_1}-\frac{ Qp_1'}{p_1^2} 
=\frac1{p_1}(Q'-\langle H\rangle_{\phi}\;p'_1)$$
where
$Q'=\langle H^2 \cos2\varepsilon H\rangle_\psi.$ 
Evaluated at $\varepsilon \rightarrow 0^+$ this becomes 
$$\langle H\rangle_{\phi}'(0)=\frac12 (\langle H^2 \rangle_\psi -  \langle H\rangle_\psi^2) =: \frac{\sigma_H^2}2 \geq 0 $$
as claimed.  

Now consider $\langle H\rangle'_{\phi}(\varepsilon^*)$. This quantity is $\leq 0$ if and only if $Q'\leq  \langle H\rangle_{\phi}p'_1$, or equivalently (by definition) if 
$\text{Cov}(H,\cos(2\varepsilon H)) \leq 0$, 
which is often true though not always guaranteed; roughly speaking, counterexample are those where the cost function grows faster than the function $\cos(2\varepsilon h(x))$ goes to zero as $h\rightarrow h^*$.
Indeed, a sufficient condition is that the function $h(x)\cos(2\varepsilon h(x))$ is nonincreasing. 

For such cases where this holds such that $\langle H\rangle'_{\phi}(\varepsilon^*)<0$, the condition $\langle H\rangle'_{\phi}(0)> 0$ implies there must exist an optimal $\varepsilon \in (0,\varepsilon^*))$, i.e. there exists a positive $\varepsilon<\varepsilon^*$ such that $\langle 
H\rangle_{\phi}(\varepsilon) \geq \langle H\rangle_{\phi}(\varepsilon^*)  $ as claimed. When this condition fails then optimal value is $\varepsilon^*$. 

As should be evident, the exact location of the optimal $\varepsilon$ depends on the particular state $\ket{\psi}$ and cost Hamiltonian $H$. 
It can be obtained by solving $Q'(\varepsilon)=\langle H\rangle_\phi (\varepsilon) p'_1$ for $\varepsilon$; we do not attempt to derive further expressions here. In practice as explained one may search over or variationally optimize the value of $\varepsilon$. For multiple steps the state and hence the optimal $\varepsilon$ changes over each step which further motivates the adaptive approaches discussed in the text.  
  
\end{document}